\documentclass[draft]{birkjour}
\allowdisplaybreaks 
\usepackage[noadjust]{cite}
\usepackage{xcolor}
\RequirePackage[all]{xy}

\newtheorem{thm}{Theorem}[section]

\newtheorem{lem}[thm]{Lemma}

\theoremstyle{definition}
\newtheorem{defn}[thm]{Definition}
\theoremstyle{remark}
\newtheorem{rem}[thm]{Remark}

\newtheorem*{ex}{Example}
\numberwithin{equation}{section}

\newcommand{\BibTeX}{B\kern-0.1emi\kern-0.017emb\kern-0.15em\TeX}
\newcommand{\XYpic}{$\mathrm{X\kern-0.3em\raisebox{-0.18em}{Y}}$-$\mathrm{pic}\,$}

\def\cl{\mathcal {G}} 

\newcommand{\ed}{\end{document}}

\def\va{\bigtriangleup}
\begin{document}

%
%
%
%
%
%
%
%
%

\title[Basis-free Formulas for Characteristic Polynomial Coefficients]{Basis-free Formulas for Characteristic Polynomial Coefficients in Geometric Algebras}
\author[K. Abdulkhaev]{Kamron  Abdulkhaev}
\address{%
HSE University\\ 
101000 Moscow\\
Russia\\}
\email{ksabdulkhaev@edu.hse.ru}

\author[D. Shirokov]{Dmitry Shirokov}
\address{%
HSE University\\ 
101000 Moscow\\
Russia\\ \\
\and \\ \\
Institute for Information Transmission Problems of the Russian Academy of Sciences\\
127051 Moscow\\
Russia\\}
\email{dshirokov@hse.ru}
%

\date{\today}
\begin{abstract}
In this paper, we discuss characteristic polynomials in (Clifford) geometric algebras $\cl_{p,q}$ of vector space of dimension $n=p+q$. We present basis-free formulas for all characteristic polynomial coefficients in the cases $n\leq 6$, alongside with a method to obtain general form of these formulas. The formulas involve only the operations of geometric product, summation, and operations of conjugation. All the formulas are verified using computer calculations. We present an analytical proof of all formulas in the case $n=4$, and  one of the formulas in the case $n=5$. We present some new properties of the operations of conjugation and grade projection and use them to obtain the results of this paper. We also present formulas for characteristic polynomial coefficients in some special cases. In particular, the formulas for vectors (elements of grade $1$) and basis elements are presented in the case of arbitrary $n$, the formulas for rotors (elements of spin groups) are presented in the cases $n\leq 5$. The results of this paper can be used in different applications of geometric algebras in computer graphics, computer vision, engineering, and physics. The presented basis-free formulas for characteristic polynomial coefficients can also be used in symbolic computation.
\end{abstract}
\label{page:firstblob}
\maketitle
\section{Introduction}
In this paper, we discuss characteristic polynomials in (Clifford) geometric algebras $\cl_{p,q}$, $p+q=n\geq 1$. We solved the problem of obtaining basis-free formulas for characteristic polynomial coefficients for the cases $n\leq6$. These formulas involve only the operations of geometric product, summation, and the operations of conjugation (the grade involution, the reversion, and one additional operation of conjugation $\va$). 

\vspace{-0.3pt} 

This paper is an extended version of the short note in Conference Proceedings \cite{article-short_version-cgi2021}, where the cases $n\leq 5$ are considered. In the present paper, for the first time, we introduce the formulas for the characteristic polynomial coefficients in geometric algebras in the case $n=6$. Also, we present a  method to obtain general form of these formulas using basis-free form of one of the coefficients (determinant). For the first time, we present formulas for the characteristic polynomial coefficients in some special cases. In particular, we present the formulas for vectors (elements of grade~$1$) and basis elements in the case of arbitrary $n$. We present simplification of the formulas for rotors (elements of spin groups ${\rm Spin}_{+}(p,q)$) in the cases $p+q=n\leq 5$. Lemmas \ref{thm:ch_poly_e} and \ref{lemma:uk_show} and Theorems \ref{lemma:c_2s-1=0}, \ref{thm:like_basis_el}, \ref{thm:spin}, and \ref{thm:cross5} are new. 

\vspace{-0.2pt} 

This paper is organized as follows. In Section \ref{section2}, we present some new properties of the operations of conjugation and grade projection and use them to obtain the results of this paper. In Section \ref{section3}, we discuss the notion of characteristic polynomial in geometric algebras and remind the recursive formulas for characteristic polynomial coefficients from \cite{Shirokov}. In Section \ref{section4}, we present an analytic proof of the basis-free formulas for characteristic polynomial coefficients in the case $n=4$. In Section \ref{section5}, we solve the same problem in the case $n=5$. The formulas are verified using symbolic computation. We present an analytical proof of one of the formulas. In Section \ref{section6}, we present formulas for the characteristic polynomial coefficients of elements with some specific conditions on their  powers. In particular, we obtain the basis-free formulas for vectors and basis elements in the case of arbitrary $n$, and rotors in the cases $n\leq5$. In Section \ref{section7}, we introduce a method to obtain a general form of basis-free formulas for characteristic polynomial coefficients. We illustrate the method for the cases $n\leq 5$. In Section \ref{section8}, we apply the method for the case $n=6$. Using numerical Geometric Algebra package for Python~\cite{Python}, we checked that the formulas give valid results for geometric algebra elements with random integer coefficients. Some of the formulas are moved to Appendix \ref{appendix} because of their cumbersomeness.

\vspace{-0.2pt} 

The geometric algebras of vector spaces of dimensions $n=4,5,$ and $6$ are important for different applications. The space-time algebra $\cl_{1,3}$ is widely used for applications in physics \cite{Doran,Hestenes,Lasenby}, the conformal geometric algebra $\cl_{4,1}$ is widely used in computer science and engineering \cite{Bayro,Breuils,Dorst,Hildenbrand,Li}, the geometric algebra $\cl_{3,3}$ of projective geometry is used in computer vision and computer graphics \cite{Dorst2,Klawitter}, the conformal space-time algebras $\cl_{4,2}$ and $\cl_{2,4}$ are used in physics \cite{Dirac,Doran}.

\vspace{-0.2pt} 

The characteristic polynomial and related concepts (eigenvectors, eigenvalues) are widely used in computer vision (see, for example, on eigenfaces and the computer vision problem of human face recognition \cite{eigenface3,eigenface1,eigenface2}). The characteristic polynomial coefficients are used to solve the Sylvester and Lyapunov equations in geometric algebra \cite{CGI20,CGI20_extend}. The presented basis-fee formulas for characteristic polynomial coefficients can also be used in symbolic computation using different software \cite{AcusPack,Python,HitzerMatlab,Python2}.

\section{Grade Projections and Operations of Conjugation in Geometric Algebras}\label{section2}

Let us consider the (Clifford) geometric algebra $\cl_{p,q}$, $p+q=n$ \cite{Hestenes,Lasenby,Lounesto} with the generators $e_1$, $e_2$, \ldots, $e_n$ and the identity element $e$. The generators satisfy the conditions
\begin{eqnarray}
    e_ae_b+e_be_a=2\eta_{ab}e,\qquad a,b=1,\ldots,n, \nonumber
\end{eqnarray}
where $\eta =(\eta_{ab})={\rm diag}(1,\ldots, 1, -1, \ldots , -1)$ is the diagonal matrix with its first $p$ entries equal to $1$ and the last $q$ entries equal to $-1$ on the diagonal. 

We call the subspace of $\cl_{p,q}$ of elements, which are linear combinations of the basis elements 
\begin{eqnarray}
    e_{a_1 \ldots a_k}:= e_{a_1}\cdots e_{a_k},\qquad
    a_1 < a_2 < \cdots < a_k,\qquad
    k = 0, 1, \ldots, n,\label{def:basis}
\end{eqnarray}
with multi-indices of length $k$, the subspace of grade $k$ and denote it by $\cl_{p,q}^k$. Elements of grade 0 are identified with scalars $\cl_{p,q}^0  \equiv \mathbb{R}, e \equiv 1$.
The projection of any element $U\in\cl_{p,q}$ onto the subspace $\cl^k_{p,q}$ is denoted by $\langle U\rangle_k$ (or $U_k$ to simplify notation) in this paper.
We have
\begin{eqnarray}
    \langle U+V \rangle_k =\langle U \rangle_k+\langle V \rangle_k,\quad
    \langle\lambda U \rangle_k=\lambda \langle U \rangle_k,\quad
\lambda \in \mathbb{R},\quad
    U,V\in \cl_{p,q}.\label{proj}
\end{eqnarray}
An arbitrary element $U\in\cl_{p,q}$ can be written in the form
\begin{eqnarray}
        U = \sum_{k=0}^n\langle U\rangle_k,\qquad
    \langle U\rangle_k \in \cl_{p,q}^k.\label{conjug}
\end{eqnarray}
The scalar $\langle U \rangle_0$ is called the scalar part of $U$. We have the property
\begin{eqnarray}
\langle UV \rangle_0=\langle VU \rangle_0,\qquad \forall U, V\in\cl_{p,q}.\label{comm}
\end{eqnarray}

\begin{defn}[\cite{Shirokov}]
\normalfont
Any operation of the form
\begin{eqnarray}
    U \mapsto
    \sum_{k=0}^n\lambda_k\langle U \rangle_k,\qquad
    \lambda_k = \pm 1\qquad
\end{eqnarray}
\normalfont
is called \textit{an operation of conjugation} in  $\cl_{p,q}$.
\end{defn}
Note that the operation of conjugation is an involution: the square of each operation equals the identical operation. 
The operations of conjugation commute with each other.
We have three classical operations of conjugation: the grade involution, the reversion, and the Clifford conjugation\footnote{The Clifford conjugation is a superposition of the grade involution $\widehat{~~}$ and the reversion $\widetilde{~~}$. Note that some authors \cite{Lounesto} denote the Clifford conjugation by $\stackrel{\overline\quad}{\quad}$. We do not use separate notation for the Clifford conjugation in this paper and write the combination of the two symbols $\widehat{~~}$ and $\widetilde{~~}$.}:
\begin{eqnarray}
    \widehat{U}=\sum_{k=0}^n (-1)^k\langle U \rangle_k,\,\,\,\,
    \widetilde{U}=\sum_{k=0}^n (-1)^\frac{k(k-1)}{2}\langle U \rangle_k,\,\,\,\,
    \widehat{\widetilde{U}}=\sum_{k=0}^n (-1)^\frac{k(k+1)}{2}\langle U \rangle_k.\label{conjug2}
\end{eqnarray}
These operations have the following properties
\begin{eqnarray}
\label{conj_prop}
    \widehat{UV}=\widehat{U}\widehat{V},\qquad
    \widetilde{UV}=\widetilde{V}\widetilde{U},\qquad
    \widehat{\widetilde{UV}}=\widehat{\widetilde{V}}\widehat{\widetilde{U}},\qquad
    \forall\, U,V \in\cl_{p,q}.
\end{eqnarray}
\begin{defn}[\cite{Shirokov}]\normalfont
We call an operation of conjugation of the form
\begin{eqnarray}
    U^\bigtriangleup=\sum_{k=0}^n (-1)^\frac{k(k-1)(k-2)(k-3)}{24}\langle U \rangle_k
\end{eqnarray}
\textit{an additional operation of conjugation} in $\cl_{p,q}$ (or \textit{$\va$-conjugation}).
\end{defn}
Note that we have $(UV)^\bigtriangleup\neq U^\bigtriangleup V^\bigtriangleup$ and $(UV)^\bigtriangleup\neq V^\bigtriangleup U^\bigtriangleup $
in the general case. However, the operation $\va$ has the following weaker property by  Lemma~\ref{lemma:grade0} and~(\ref{comm}):
\begin{eqnarray}
\langle (U V)^\va \rangle_0=\langle U^\va V^\va \rangle_0=\langle V^\va U^\va \rangle_0,\quad
\forall U,V\in\cl_{p,q}.
\end{eqnarray}
We widely use the operation $\va$ in this paper.

We need the following two lemmas to prove the results of this paper.

\begin{lem}\label{lemma:grade0} We have the following properties
 \begin{eqnarray}
 \langle UV^\star \rangle_0&=&\langle U^\star V \rangle_0,\quad
 \forall U,V\in\cl_{p,q},
\label{conj_prop2}\\
\langle U^\star V^\star \rangle_0&=&\langle U V \rangle_0,\quad
 \forall U,V\in\cl_{p,q},\label{conj_prop22}
 \end{eqnarray}
 where $\star$ is any operation of conjugation (\ref{conjug});
 \begin{eqnarray}
 \langle U\rangle_0&=&\langle U^\bullet  \rangle_0,\quad
 \forall U\in\cl_{p,q},
\label{conj_prop3}\\
\langle U^\bullet V^\bullet \rangle_0&=&\langle (U V)^\bullet \rangle_0,\quad
\forall U,V\in\cl_{p,q},\label{conj_prop33}
\end{eqnarray}
 where $\bullet$ is any operation of conjugation (\ref{conjug}) that does not change the sign of grade 0 (i.e. $\lambda_0=+1$).
\end{lem}
\begin{proof}
If
\begin{eqnarray}
        U&=&X+x,\qquad\mbox{where}\qquad (X+x)^\star=X-x, \nonumber\\
        V&=&Y+y,\qquad\mbox{where}\qquad (Y+y)^\star=Y-y, \nonumber
\end{eqnarray}
then
\begin{eqnarray}
    \langle Xy \rangle_0=\langle Yx \rangle_0=0,
\label{eqn:lemma_grade0_xy}
\end{eqnarray}
because the elements $X$ and $y$ are of different grades (similarly for the elements $Y$ and $x$) by construction. Using (\ref{eqn:lemma_grade0_xy}) and (\ref{proj}), we get
\begin{eqnarray}
    \langle UV^\star \rangle_0&=&
    \langle (X+x)(Y-y) \rangle_0=
    \langle XY-Xy+xY-xy \rangle_0\nonumber\\
    &=&
    \langle XY-xy \rangle_0
    \label{eqn:lemma_grade0_uv^s},\\
    \langle U^\star V\rangle_0&=&
    \langle (X-x)(Y+y) \rangle_0=
    \langle XY+Xy-xY-xy \rangle_0\nonumber\\
    &=&
    \langle XY-xy \rangle_0
    .    \label{eqn:lemma_grade0_u^sv}
\end{eqnarray}
From the equality of the right-hand sides of the expressions (\ref{eqn:lemma_grade0_uv^s}), (\ref{eqn:lemma_grade0_u^sv}), we get the equality of the left-hand sides (\ref{conj_prop2}).
Substituting $U^\star$ for $U$ in (\ref{conj_prop2}), we get  (\ref{conj_prop22}).
Substituting $e$ for $V$ in (\ref{conj_prop2}) and using $e^\bullet=e$, we get  (\ref{conj_prop3})\footnote{Note that the property (\ref{conj_prop3}) is also follows from the definition of $\bullet$ and the fact that grade projections commute with operations of conjugation.}.
Using (\ref{conj_prop22}) and (\ref{conj_prop3}), we get  (\ref{conj_prop33}). 
\end{proof}

In the next lemma, we discuss the relation between the operations of conjugation (left-hand sides of the equalities) and the grade projections (right-hand sides of the equalities). Note that the same left-hand sides of equalities are used in \cite{Shirokov} to realize the scalar part operation $\langle~~\rangle_0$ in the cases of smaller dimensions $n$. For example, the left-hand side of (\ref{lem23}) is equal to $4\langle U \rangle_0$ in the cases $n\leq 6$.
\begin{lem}
\label{lemma:proj_realization}
For $n\leq7$, the following equalities hold
 \begin{eqnarray} 
U+\widehat{\widetilde{U}}&=&
 	2( U_0+
 	 U_3+
 	 U_4+
 	 U_7),\label{u+u^a}\\
 U+\widehat{\widetilde{U}}+\widehat{U}^\bigtriangleup+\widetilde{U}^\bigtriangleup&=&4( U_0+ U_7),\label{lem23}\\
 	U+\widehat{U}+\widetilde{U}^\bigtriangleup+\widehat{\widetilde{U}}^\bigtriangleup&=&4( U_0+ U_6),\label{u+u^g}\\
 	U+\widetilde{U}+\widehat{U}^\bigtriangleup+\widehat{\widetilde{U}}^\bigtriangleup&=&4( U_0+ U_5),\label{u+u^r}\\
 	U+\widehat{U}+\widetilde{U}+\widehat{\widetilde{U}}&=&4( U_0+ U_4),\label{u+u^gg}
\end{eqnarray}
where the simplified notation $U_k$ is used for $\langle U \rangle_k$.
\end{lem}

\begin{proof}
    The proof is by direct calculation. For example, from (\ref{conjug}) and (\ref{conjug2}), we have 
    \begin{eqnarray}
        U= U_0+ U_1+ U_2+ U_3+ U_4+ U_5+ U_6+ U_7,\label{lemma_u}\\
        \widehat{\widetilde{U}}= U_0- U_1- U_2+ U_3+ U_4- U_5- U_6+ U_7.\label{lemma_u^a}
    \end{eqnarray}
    Summing (\ref{lemma_u}) and (\ref{lemma_u^a}), we obtain the expression (\ref{u+u^a}) from Lemma \ref{lemma:proj_realization}
    \begin{eqnarray}
     	U+\widehat{\widetilde{U}}
     	=
     	2 U_0+
     	2 U_3+
     	2 U_4+
     	2 U_7 
     	.\nonumber
    \end{eqnarray}
    Similarly one can prove all the other equalities from Lemma \ref{lemma:proj_realization}.
\end{proof}

\section{Characteristic Polynomials in Geometric Algebras}\label{section3}

Characteristic polynomials in geometric algebras $\cl_{p,q}$, $n=p+q$, are discussed in \cite{Helm} and \cite{Shirokov}. We use the notation $N:= 2^{[\frac{n+1}{2}]}$, where square brackets mean taking the integer part.

\begin{defn}[\cite{Shirokov}]\label{def:ch_poly}\normalfont
    Let us consider an arbitrary element $U\in\cl_{p,q}$. We call \textit{the characteristic polynomial} of $U$
    \begin{eqnarray}
        \varphi_U(\lambda) &:=& {\rm det}(\beta(\lambda e - U))={\rm Det}(\lambda e - U) \nonumber\\
        &=&\lambda^N-C_{(1)}\lambda^{N-1}-\cdots-C_{(N-1)}\lambda-C_{(N)}\in \cl_{p,q}^0,\label{ch_poly}
    \end{eqnarray}
    where $C_{(j)} = C_{(j)}(U)=c_{(j)}(\beta(U)) \in \cl_{p,q}^0\equiv \mathbb{R}, j=1,\ldots,N$ can be interpreted as constants or as elements of grade $0$ and are called \textit{characteristic polynomial coefficients} of $U$. Here $c_{(j)}(\beta(U))$ are the ordinary characteristic polynomial coefficients of the matrix $\beta(U)$ and
        $$\beta:\cl_{p,q}\to \beta(\cl_{p,q})\subset 
        M_{p,q},$$
        where
        $$M_{p,q}:=\left\lbrace
\begin{array}{ll}
{\rm Mat}(2^{\frac{n}{2}}, {\mathbb C}), & \mbox{if $n$ is even,}\\
{\rm Mat}(2^{\frac{n-1}{2}}, {\mathbb C})\oplus{\rm Mat}(2^{\frac{n-1}{2}}, {\mathbb C}), & \mbox{if $n$ is odd,}
\end{array}
\right.\label{beta}$$
is a representation of $\cl_{p,q}$ (of not minimal dimension, see the details in \cite{Shirokov}).
\end{defn}

Note that the trace ${\rm Tr}(U):={\rm tr}(\beta(U))=N \langle U \rangle_0=C_{(1)}$ and the determinant ${\rm Det}(U):=\det(\beta(U))=-C_{(N)}$ are particular cases of characteristic polynomial coefficients.  
The basis-free formulas for the determinant allow us to calculate the adjugate ${\rm Adj}(U)$ and the inverse $U^{-1}$ in $\cl_{p,q}$ (see \cite{Acus,Hitzer,Shirokov}).

We use the following recursive formulas for the characteristic polynomial coefficients $C_{(k)}$, $k=1, \ldots, N$, $N=2^{[\frac{n+1}{2}]}$ from \cite{Shirokov}. The elements $U_{(k)}\in\cl_{p,q}$, $k=1, \ldots, N$, are auxiliary.

\begin{thm}[\cite{Shirokov}]
\label{theorem:mainTheorem} Let us consider an arbitrary element $U\in\cl_{p,q}$, $n=p+q$, $N=2^{[\frac{n+1}{2}]}$. Setting $U_{(1)}=U$, we have
\begin{eqnarray}
&&U_{(k+1)}=U(U_{(k)}-C_{(k)}),\quad C_{(k)}=\frac{N}{k}\langle U_{(k)}\rangle_0,\,\,\, k=1, \ldots, N,\label{recur}\\ &&{\rm Det}(U)=-U_{(N)}=-C_{(N)}=U(C_{(N-1)}-U_{(N-1)}),\\
&&{\rm Adj}(U)=C_{(N-1)}-U_{(N-1)},\qquad U^{-1}=\frac{{\rm Adj}(U)}{{\rm Det}(U)}.\label{main_det}
\end{eqnarray}
\end{thm}

\section{The Cases $n \leq 4$}\label{section4}

The basis-free formulas for all characteristic polynomial coefficients in $\cl_{p,q}$, $n=p+q\leq 4$ were presented in \cite{Shirokov}. These formulas were obtained using the algorithm from Theorem \ref{theorem:mainTheorem}. The formulas
(\ref{c2_4}) and (\ref{c3_4}) were proved in \cite{Shirokov} using computer calculations, all other formulas from Theorem \ref{theorem2} were proved analytically. We present an analytic proof of the formulas (\ref{c2_4}) and (\ref{c3_4}).

\begin{thm}\label{theorem2} In the cases $n = 1, 2, 3, 4$, we have the following basis-free formulas for the characteristic polynomial coefficients $C_{(k)}\in\cl^0_{p,q}$, $k=1, 2, \ldots, N$:
\begin{eqnarray}
n=1,\qquad &&C_{(1)}=U+\widehat{U},\qquad
    C_{(2)}=-U\widehat{U};\nonumber\\
n=2,\qquad &&C_{(1)}=U+\widehat{\widetilde{U}},\qquad  C_{(2)}=-U\widehat{\widetilde{U}};\nonumber\\
n=3,\qquad  &&C_{(1)}=
    U+\widehat{U}+\widetilde{U}+\widehat{\widetilde{U}},\nonumber\\
    &&C_{(2)}=-(
    U\widetilde{U}+
    U\widehat{U}+
    U\widehat{\widetilde{U}}+
    \widehat{U}\widehat{\widetilde{U}}+
    \widetilde{U}\widehat{\widetilde{U}}+
    \widehat{U}\widetilde{U}),\nonumber\\
   && C_{(3)}=
    U\widehat{U}\widetilde{U}+
    U\widehat{U}\widehat{\widetilde{U}}+
    U\widetilde{U}\widehat{\widetilde{U}}+
    \widehat{U}\widetilde{U}\widehat{\widetilde{U}},\nonumber\\
    &&C_{(4)}=-
    U\widehat{U}\widetilde{U}\widehat{\widetilde{U}};\nonumber\\
n=4,\qquad  &&C_{(1)}=
    U+\widehat{\widetilde{U}}+\widehat{U}^\bigtriangleup+\widetilde{U}^\bigtriangleup,\label{c1_4}\\
    &&C_{(2)}=-(
    U\widehat{\widetilde{U}}+
    U\widehat{U}^\bigtriangleup+
    U\widetilde{U}^\bigtriangleup\nonumber\\
    &&\qquad\quad+\widehat{\widetilde{U}}\widehat{U}^\bigtriangleup+
    \widehat{\widetilde{U}}\widetilde{U}^\bigtriangleup+
    (\widehat{U}\widetilde{U})^\bigtriangleup),\label{c2_4}\\
    &&C_{(3)}=
    U\widehat{\widetilde{U}}\widehat{U}^\bigtriangleup+
    U\widehat{\widetilde{U}}\widetilde{U}^\bigtriangleup+
    U(\widehat{U}\widetilde{U})^\bigtriangleup+
    \widehat{\widetilde{U}}(\widehat{U}\widetilde{U})^\bigtriangleup,\label{c3_4}\\
    &&C_{(4)}=-
    U\widehat{\widetilde{U}}(\widehat{U}\widetilde{U})^\bigtriangleup.\label{c4_4}
 \end{eqnarray}
\end{thm}

\begin{proof}(of (\ref{c2_4})) Our analytical proof of the formula (\ref{c2_4}) for $C_{(2)}$ in the case $n=4$ is in two steps. Step 1: we prove that the projection of the expression (\ref{c2_4}) onto the subspace of grade $0$ is equal to $C_{(2)}$ from (\ref{recur}). Step 2: we prove that the expression (\ref{c2_4}) belongs to $\cl^0_{p,q}$.

Step 1: Using (\ref{recur}), (\ref{c1_4}), and (\ref{proj}), we get
\begin{eqnarray}
     &&C_{(2)}=2\langle U(U-C_{(1)})\rangle_0=-2\langle U(\widehat{\widetilde{U}}+\widehat{U}^\bigtriangleup+\widetilde{U}^\bigtriangleup)\rangle_0\nonumber\\
    &&=-\langle U\widehat{\widetilde{U}}\rangle_0
    -\langle U\widehat{U}^\bigtriangleup\rangle_0
    -\langle U\widetilde{U}^\bigtriangleup\rangle_0
    -\langle U\widehat{\widetilde{U}}\rangle_0
    -\langle U\widehat{U}^\bigtriangleup\rangle_0
    -\langle U\widetilde{U}^\bigtriangleup\rangle_0.\nonumber
\end{eqnarray}
Using the properties (\ref{conj_prop2}) and (\ref{conj_prop3}) for the operations $\,\,\widehat{}\,\,$, $\,\,\widetilde{}\,\,$, and $\va$, we get
$$\langle U\widehat{\widetilde{U}} \rangle_0 = \langle\widehat{U}\widetilde{U}\rangle_0=
\langle(\widehat{U}\widetilde{U})^\bigtriangleup\rangle_0,\quad
\langle U\widehat{U}^\bigtriangleup\rangle_0
=\langle \widehat{\widetilde{U}}\widetilde{U}^\bigtriangleup \rangle_0,\quad \langle U\widetilde{U}^\bigtriangleup\rangle_0=\langle \widehat{\widetilde{U}}\widehat{U}^\bigtriangleup \rangle_0.
$$
Finally, we obtain
\begin{eqnarray}
C_{(2)}=-\langle
    U\widehat{\widetilde{U}}+
    U\widehat{U}^\bigtriangleup+
    U\widetilde{U}^\bigtriangleup+
    \widehat{\widetilde{U}}\widehat{U}^\bigtriangleup+
    \widehat{\widetilde{U}}\widetilde{U}^\bigtriangleup+
    (\widehat{U}\widetilde{U})^\bigtriangleup\rangle_0,\label{grade0_c2_4}
\end{eqnarray}
which differs from (\ref{c2_4}) only by the scalar part operation.

Step 2: Using the properties (\ref{conj_prop}), we conclude that the expression $U\widehat{\widetilde{U}}+
    (\widehat{U}\widetilde{U})^\bigtriangleup$ does not change under the operations $\widehat{\widetilde{~~}}$ and $\widehat{\bigtriangleup}$:
\begin{eqnarray}
    U\widehat{\widetilde{U}}+
    (\widehat{U}\widetilde{U})^\bigtriangleup=
    (U\widehat{\widetilde{U}}+
    (\widehat{U}\widetilde{U})^\bigtriangleup)^{\widehat{\widetilde{~~}}}=
    (U\widehat{\widetilde{U}}+
    (\widehat{U}\widetilde{U})^\bigtriangleup)^{\widehat{\bigtriangleup}}
    .\label{c2_4_a}
\end{eqnarray}
This means that the sum of the first and last terms of (\ref{c2_4}) belongs to the subspace of grade $0$:
\begin{eqnarray}
U\widehat{\widetilde{U}}+
    (\widehat{U}\widetilde{U})^\bigtriangleup \in \cl^0_{p,q}.\label{rr}
\end{eqnarray}
For the other four terms of (\ref{c2_4}), using (\ref{u+u^a}), we get\footnote{The commutator and anticommutator of two arbitrary elements $U,V \in \cl_{p,q}$ are denoted by $[U,V]=UV-VU$ and $\{U,V\}=UV+VU$ respectively.}\,\footnote{We remind that we use the simplified notation $U_k:=\langle U \rangle_k$ in this paper.}
\begin{eqnarray}
    &&U\widehat{U}^\bigtriangleup+
    \widehat{\widetilde{U}}\widetilde{U}^\bigtriangleup+
    U\widetilde{U}^\bigtriangleup+
    \widehat{\widetilde{U}}\widehat{U}^\bigtriangleup
    =
    (U+\widehat{\widetilde{U}})(U+\widehat{\widetilde{U}})^{\widetilde{\bigtriangleup}}\nonumber\\
    &&=
    4(U_0+U_3+U_4)(U_0-U_3-U_4)\nonumber\\
    &&=
    4(U_0^2-U_3^2-U_4^2-\{U_3,U_4\}-[U_0,U_3+U_4])
    \nonumber\\
    &&=
    4(U_0^2-U_3^2-U_4^2)\in\cl_{p,q}^0,\label{c2_4_b}
\end{eqnarray}
because $U_0^2, U_3^2, U_4^2\in\cl_{p,q}^0$ and $\{U_3,U_4\}=0$ (see, for example, \cite{Shirokov_2009}).
Summing (\ref{rr}) and (\ref{c2_4_b}), we conclude that the expression (\ref{c2_4}) belongs to $\cl_{p,q}^0$.
\end{proof}

\begin{proof}(of (\ref{c3_4})) The analytical proof of the formula (\ref{c3_4}) for $C_{(3)}$ is in two steps.

Step 1: Using the properties (\ref{comm}), (\ref{conj_prop2}), and (\ref{conj_prop3}) for the operations $\widehat{~~}$, $\widetilde{~~}$, and $\bigtriangleup$, we get 
\begin{eqnarray}
    \langle 
    U\widehat{\widetilde{U}}\widehat{U}^\bigtriangleup
    \rangle_0 
    =
    \langle 
    U\widehat{\widetilde{U}}\widetilde{U}^\bigtriangleup
    \rangle_0 
    =
    \langle 
    U(\widehat{U}\widetilde{U})^\bigtriangleup
    \rangle_0
    =
    \langle
    \widehat{\widetilde{U}}(\widehat{U}\widetilde{U})^\bigtriangleup
    \rangle_0
    .\label{prop_c3_4_uuu}
\end{eqnarray}
Using (\ref{recur}), (\ref{proj}), (\ref{c1_4}), (\ref{c2_4}), and (\ref{prop_c3_4_uuu}), we get
\begin{eqnarray}
    &&U_{(2)}=U(U-C_{(1)})
    =-(
    U\widehat{\widetilde{U}}+
    U\widehat{U}^\bigtriangleup+
    U\widetilde{U}^\bigtriangleup
    ),\nonumber\\
    &&C_{(3)}
    =\frac{4}{3}\langle U(
    U_{(2)}-C_{(2)}
    )\rangle_0=
    \frac{4}{3}\langle 
    U\widehat{\widetilde{U}}\widehat{U}^\bigtriangleup+
    U\widehat{\widetilde{U}}\widetilde{U}^\bigtriangleup+
    U(\widehat{U}\widetilde{U})^\bigtriangleup
    \rangle_0\nonumber\\
   &&=
    \langle
    U\widehat{\widetilde{U}}\widehat{U}^\bigtriangleup+
    U\widehat{\widetilde{U}}\widetilde{U}^\bigtriangleup+
    U(\widehat{U}\widetilde{U})^\bigtriangleup+
    \widehat{\widetilde{U}}(\widehat{U}\widetilde{U})^\bigtriangleup
    \rangle_0,
    \label{c3_4_grade0_2}
\end{eqnarray}
which differs from (\ref{c3_4}) only by the scalar part operation.

Step 2: Let us prove that the expression (\ref{c3_4}) belongs to $\cl^0_{p,q}$. It can be represented in the form
\begin{eqnarray}
(\widehat{U}\widetilde{U})^{\widehat{~~}}
    (U+\widehat{\widetilde{U}})^{\widehat{\bigtriangleup}}
    +
    (U+\widehat{\widetilde{U}})(\widehat{U}\widetilde{U})^{\bigtriangleup}=\widehat{A} \widehat{B}^\bigtriangleup+
    BA^\bigtriangleup,
   \label{TT}
\end{eqnarray}
where we use the notation $A:=\widehat{U}\widetilde{U}$ and $B:=U+\widehat{\widetilde{U}}$.
Using $\widehat{\widetilde{A}}=A$ and (\ref{u+u^a}), we get
\begin{eqnarray}
A=A_0+A_3+A_4,\qquad  B=B_0+B_3+B_4,\qquad \mbox{where}\quad A_i, B_i \in \cl_{p,q}^i.\nonumber
\end{eqnarray}
Therefore, the expression (\ref{TT}) is equal to
\begin{eqnarray}
&&(A_0-A_3+A_4)(B_0-B_3-B_4)+(B_0+B_3+B_4)(A_0+A_3-A_4)\nonumber\\
&&=\{A_0,B_0\}-[A_0,B_3]-[A_0,B_4]-[A_3,B_0]+\{A_3,B_3\}+\{A_3,B_4\}\nonumber\\
&&+[A_4,B_0]-\{A_4,B_3\}-\{A_4,B_4\}=\{A_0,B_0\}+\{A_3,B_3\}-\{A_4,B_4\},\nonumber
\end{eqnarray}
which belongs to $\cl^0_{p,q}$ because $\{U_{n-1},V_{n}\}=0$ for even $n$ and the expressions $\{U_0,V_0\}$, $\{U_{n-1},V_{n-1}\}$, $\{U_n,V_n\}$ belong to $\cl_{p,q}^0$ (see, for example, \cite{Shirokov_2009}\footnote{Alternatively, we can use the quaternion type classification of Clifford algebra elements \cite{quat,quat2} to prove this.}). Therefore the expression (\ref{c3_4}) belongs to $\cl_{p,q}^0$.
\end{proof}

\section{The Case $n=5$}\label{section5}
In this section, we present basis-free formulas for all characteristic polynomial coefficients in the geometric algebras $\cl_{p,q}$, $n=p+q=5$. The formula (\ref{c8_5}) for $C_{(8)}=-{\rm Det}(U)$ is presented in \cite{Shirokov} and in some another form in \cite{Acus}. The formula for $C_{(1)}={\rm Tr}(U)$ is also presented in \cite{Shirokov}.
\begin{thm}\label{thm:ck_n=5}
In the case $n=5$, we have the following basis-free formulas for the characteristic polynomial coefficients $C_{(k)}\in\cl^0_{p,q}$, $k=1, 2, \ldots, 8$:
{\footnotesize
    \begin{eqnarray}
    C_{(1)} &=&
    U
    +
    \widehat{\widetilde{U}}
    +
    \widehat{U}
    +
    \widetilde{U}
    +
    \widehat{U}^\bigtriangleup
    +
    \widetilde{U}^\bigtriangleup
    +
    U^\bigtriangleup
    +
    \widehat{\widetilde{U}}^\bigtriangleup
    ,\label{c1_5}\\
    C_{(2)} &=&
    -(
            U\widehat{\widetilde{U}}  
            +
            U\widehat{U}  
            +
            U\widetilde{U}  
            +
            U\widehat{U}^\bigtriangleup  
            +
            U\widetilde{U}^\bigtriangleup  
            +
            UU^\bigtriangleup  
            +
            U\widehat{\widetilde{U}}^\bigtriangleup
            +
            \widehat{U}\widetilde{U}  
            +
            \widehat{\widetilde{U}}\widetilde{U}\nonumber\\
            &+&
            \widehat{\widetilde{U}}\widehat{U}
            +
            \widehat{\widetilde{U}}\widetilde{U}^\bigtriangleup
            +
            \widehat{\widetilde{U}}\widehat{U}^\bigtriangleup
            +
             \widehat{U}\widehat{U}^\bigtriangleup 
            +
            \widehat{\widetilde{U}}U^\bigtriangleup
            +
            (\widehat{U}\widetilde{U})^\bigtriangleup  
            +
            (\widehat{U}U)^\bigtriangleup 
            +
            (\widehat{U}\widehat{\widetilde{U}})^\bigtriangleup\nonumber\\
            &+&
            \widetilde{U}\widehat{\widetilde{U}}^\bigtriangleup
            +
            \widetilde{U}U^\bigtriangleup
            +
            \widetilde{U}\widetilde{U}^\bigtriangleup 
            +
             \widehat{U}\widetilde{U}^\bigtriangleup
            +
            (U\widehat{\widetilde{U}})^\bigtriangleup
            +
            (\widetilde{U}\widehat{\widetilde{U}})^\bigtriangleup
            +
            (\widetilde{U}U)^\bigtriangleup
            +
            \widehat{U}U^\bigtriangleup\nonumber\\
            &+&
            \widehat{U}\widehat{\widetilde{U}}^\bigtriangleup
            +
            \widehat{\widetilde{U}}\widehat{\widetilde{U}}^\bigtriangleup
            +
             \widetilde{U}\widehat{U}^\bigtriangleup ),\label{c2_5}\\
    C_{(3)} &=&
            U\widehat{\widetilde{U}}\widehat{U}
            +
            U\widehat{\widetilde{U}}\widetilde{U}
            +
            U\widehat{U}\widetilde{U}
            +
            \widehat{\widetilde{U}}\widehat{U}\widetilde{U}
            +
            U\widehat{\widetilde{U}}\widehat{U}^\bigtriangleup
            +
            U\widehat{\widetilde{U}}\widetilde{U}^\bigtriangleup
            +
            U\widehat{\widetilde{U}}U^\bigtriangleup
            +
            U\widehat{U}\widehat{U}^\bigtriangleup
            \nonumber\\
            &+&
            U\widehat{U}\widetilde{U}^\bigtriangleup
            +
            U\widehat{U}U^\bigtriangleup
            +
            U\widehat{U}\widehat{\widetilde{U}}^\bigtriangleup
            +
            U\widetilde{U}\widehat{U}^\bigtriangleup
            +
            \widehat{\widetilde{U}}\widehat{U}\widehat{U}^\bigtriangleup
            +
            \widehat{\widetilde{U}}\widehat{U}\widetilde{U}^\bigtriangleup
            +
            \widehat{U}\widetilde{U}\widehat{U}^\bigtriangleup
            \nonumber\\
            &+&
            \widehat{U}\widetilde{U}\widetilde{U}^\bigtriangleup
            +
            \widehat{U}\widetilde{U}U^\bigtriangleup
            +
            \widehat{U}\widetilde{U}\widehat{\widetilde{U}}^\bigtriangleup
            +
            U\widetilde{U}\widetilde{U}^\bigtriangleup
            +
            U\widetilde{U}U^\bigtriangleup
            +
            U\widetilde{U}\widehat{\widetilde{U}}^\bigtriangleup
            +
            \widehat{\widetilde{U}}\widehat{U}U^\bigtriangleup
            \nonumber\\
            &+&
            \widehat{\widetilde{U}}\widehat{U}\widehat{\widetilde{U}}^\bigtriangleup
            +
            \widehat{\widetilde{U}}\widetilde{U}\widehat{U}^\bigtriangleup
            +
            \widehat{\widetilde{U}}\widetilde{U}\widetilde{U}^\bigtriangleup
            +
            \widehat{\widetilde{U}}\widetilde{U}U^\bigtriangleup
            +
            \widehat{\widetilde{U}}\widetilde{U}\widehat{\widetilde{U}}^\bigtriangleup
            +
            U\widehat{\widetilde{U}}\widehat{\widetilde{U}}^\bigtriangleup
            +
            \widehat{U}(\widehat{U}U)^\bigtriangleup
            \nonumber\\
            &+&
            U(\widehat{U}\widetilde{U})^\bigtriangleup
            +
            U(\widehat{U}U)^\bigtriangleup
            +
            U(\widehat{U}\widehat{\widetilde{U}})^\bigtriangleup
            +
            U(\widetilde{U}U)^\bigtriangleup
            +
            U(\widetilde{U}\widehat{\widetilde{U}})^\bigtriangleup
            +
            U(U\widehat{\widetilde{U}})^\bigtriangleup
            \nonumber\\
            &+&
            \widehat{\widetilde{U}}(\widehat{U}\widetilde{U})^\bigtriangleup
            +
            \widehat{\widetilde{U}}(\widehat{U}U)^\bigtriangleup
            +
            \widehat{\widetilde{U}}(\widehat{U}\widehat{\widetilde{U}})^\bigtriangleup
            +
            \widehat{\widetilde{U}}(\widetilde{U}U)^\bigtriangleup
            +
            \widehat{\widetilde{U}}(\widetilde{U}\widehat{\widetilde{U}})^\bigtriangleup
            +
            \widehat{\widetilde{U}}(U\widehat{\widetilde{U}})^\bigtriangleup
            \nonumber\\
            &+&
            \widehat{U}(\widehat{U}\widetilde{U})^\bigtriangleup
            +
            \widehat{U}(\widehat{U}\widehat{\widetilde{U}})^\bigtriangleup
            +
            \widehat{U}(\widetilde{U}U)^\bigtriangleup
            +
            \widehat{U}(\widetilde{U}\widehat{\widetilde{U}})^\bigtriangleup
            +
            \widehat{U}(U\widehat{\widetilde{U}})^\bigtriangleup
            +
            \widetilde{U}(\widehat{U}\widetilde{U})^\bigtriangleup
            \nonumber\\
            &+&
            \widetilde{U}(\widehat{U}U)^\bigtriangleup
            +
            \widetilde{U}(\widehat{U}\widehat{\widetilde{U}})^\bigtriangleup
            +
            \widetilde{U}(\widetilde{U}U)^\bigtriangleup
            +
            \widetilde{U}(\widetilde{U}\widehat{\widetilde{U}})^\bigtriangleup
            +
            \widetilde{U}(U\widehat{\widetilde{U}})^\bigtriangleup
            +
            (\widehat{U}\widetilde{U}U)^\bigtriangleup
            \nonumber\\
            &+&
            (\widehat{U}\widetilde{U}\widehat{\widetilde{U}})^\bigtriangleup
            +
            (\widehat{U}U\widehat{\widetilde{U}})^\bigtriangleup
            +
            (\widetilde{U}U\widehat{\widetilde{U}})^\bigtriangleup
    ,\label{c3_5}\\
    C_{(4)} &=&
    -(
            U\widehat{\widetilde{U}}\widehat{U}\widetilde{U}
            +
            U\widehat{\widetilde{U}}\widehat{U}\widehat{U}^\bigtriangleup
            +
            U\widehat{\widetilde{U}}\widehat{U}\widetilde{U}^\bigtriangleup
            +
            U\widehat{\widetilde{U}}\widehat{U}U^\bigtriangleup
            +
            U\widehat{\widetilde{U}}\widehat{U}\widehat{\widetilde{U}}^\bigtriangleup
            +
            U\widehat{\widetilde{U}}\widetilde{U}\widehat{U}^\bigtriangleup
            \nonumber\\
            &+&
            U\widehat{\widetilde{U}}\widetilde{U}\widetilde{U}^\bigtriangleup
            +
            U\widehat{\widetilde{U}}\widetilde{U}U^\bigtriangleup
            +
            U\widehat{\widetilde{U}}\widetilde{U}\widehat{\widetilde{U}}^\bigtriangleup
            +
            U\widehat{U}\widetilde{U}\widehat{U}^\bigtriangleup
            +
            U\widehat{U}\widetilde{U}\widetilde{U}^\bigtriangleup
            +
            U\widehat{U}\widetilde{U}U^\bigtriangleup
            \nonumber\\
            &+&
            \widehat{\widetilde{U}}\widehat{U}\widetilde{U}\widetilde{U}^\bigtriangleup
            +
            \widehat{\widetilde{U}}\widehat{U}\widetilde{U}U^\bigtriangleup
            +
            \widehat{\widetilde{U}}\widehat{U}\widetilde{U}\widehat{\widetilde{U}}^\bigtriangleup
            +
            U\widehat{U}\widetilde{U}\widehat{\widetilde{U}}^\bigtriangleup
            +
            \widehat{\widetilde{U}}\widehat{U}(\widehat{U}\widetilde{U})^\bigtriangleup
            +
            \widehat{\widetilde{U}}\widehat{U}(\widehat{U}U)^\bigtriangleup
            \nonumber\\
            &+&
            U\widehat{\widetilde{U}}(\widetilde{U}U)^\bigtriangleup
            +
            U\widehat{\widetilde{U}}(\widetilde{U}\widehat{\widetilde{U}})^\bigtriangleup
            +
            U\widehat{\widetilde{U}}(U\widehat{\widetilde{U}})^\bigtriangleup
            +
            U\widehat{\widetilde{U}}(\widehat{U}\widetilde{U})^\bigtriangleup
            +
            U\widehat{\widetilde{U}}(\widehat{U}U)^\bigtriangleup
            \nonumber\\
            &+&
            U\widehat{\widetilde{U}}(\widehat{U}\widehat{\widetilde{U}})^\bigtriangleup
            +
            U\widehat{U}(\widehat{U}\widetilde{U})^\bigtriangleup
            +
            U\widehat{U}(\widehat{U}U)^\bigtriangleup
            +
            U\widehat{U}(\widehat{U}\widehat{\widetilde{U}})^\bigtriangleup
            +
            U\widehat{U}(\widetilde{U}U)^\bigtriangleup
            \nonumber\\
            &+&
            U\widehat{U}(\widetilde{U}\widehat{\widetilde{U}})^\bigtriangleup
            +
            \widehat{\widetilde{U}}\widehat{U}(\widehat{U}\widehat{\widetilde{U}})^\bigtriangleup
            +
            U\widehat{U}(U\widehat{\widetilde{U}})^\bigtriangleup
            +
            U\widetilde{U}(\widehat{U}\widetilde{U})^\bigtriangleup
            +
            U\widetilde{U}(\widehat{U}U)^\bigtriangleup
            \nonumber\\
            &+&
            U\widetilde{U}(\widehat{U}\widehat{\widetilde{U}})^\bigtriangleup
            +
            U\widetilde{U}(\widetilde{U}U)^\bigtriangleup
            +
            U\widetilde{U}(\widetilde{U}\widehat{\widetilde{U}})^\bigtriangleup
            +
            U\widetilde{U}(U\widehat{\widetilde{U}})^\bigtriangleup
            +
            U(\widehat{U}\widetilde{U}U)^\bigtriangleup
            \nonumber\\
            &+&
            U(\widehat{U}\widetilde{U}\widehat{\widetilde{U}})^\bigtriangleup
            +
            U(\widehat{U}U\widehat{\widetilde{U}})^\bigtriangleup
            +
            U(\widetilde{U}U\widehat{\widetilde{U}})^\bigtriangleup
            +
            \widehat{\widetilde{U}}\widehat{U}\widetilde{U}\widehat{U}^\bigtriangleup
            +
            \widehat{\widetilde{U}}\widehat{U}(\widetilde{U}U)^\bigtriangleup
            \nonumber\\
            &+&
            \widehat{\widetilde{U}}\widehat{U}(\widetilde{U}\widehat{\widetilde{U}})^\bigtriangleup
            +
            \widehat{\widetilde{U}}\widehat{U}(U\widehat{\widetilde{U}})^\bigtriangleup
            +
            \widehat{\widetilde{U}}\widetilde{U}(\widehat{U}\widetilde{U})^\bigtriangleup
            +
            \widehat{\widetilde{U}}\widetilde{U}(\widehat{U}U)^\bigtriangleup
            +
            \widehat{\widetilde{U}}\widetilde{U}(\widehat{U}\widehat{\widetilde{U}})^\bigtriangleup
            \nonumber\\
            &+&
            \widehat{\widetilde{U}}\widetilde{U}(\widetilde{U}U)^\bigtriangleup
            +
            \widehat{\widetilde{U}}\widetilde{U}(\widetilde{U}\widehat{\widetilde{U}})^\bigtriangleup
            +
            \widehat{\widetilde{U}}\widetilde{U}(U\widehat{\widetilde{U}})^\bigtriangleup
            +
            \widehat{\widetilde{U}}(\widehat{U}\widetilde{U}U)^\bigtriangleup
            +
            \widehat{\widetilde{U}}(\widehat{U}\widetilde{U}\widehat{\widetilde{U}})^\bigtriangleup
            \nonumber\\
            &+&
            \widehat{\widetilde{U}}(\widehat{U}U\widehat{\widetilde{U}})^\bigtriangleup
            +
            \widehat{\widetilde{U}}(\widetilde{U}U\widehat{\widetilde{U}})^\bigtriangleup
            +
            \widehat{U}\widetilde{U}(\widehat{U}\widetilde{U})^\bigtriangleup
            +
            \widehat{U}\widetilde{U}(\widehat{U}U)^\bigtriangleup
            +
            \widehat{U}\widetilde{U}(\widehat{U}\widehat{\widetilde{U}})^\bigtriangleup
            \nonumber\\
            &+&
            \widehat{U}\widetilde{U}(\widetilde{U}U)^\bigtriangleup
            +
            \widehat{U}\widetilde{U}(\widetilde{U}\widehat{\widetilde{U}})^\bigtriangleup
            +
            \widehat{U}\widetilde{U}(U\widehat{\widetilde{U}})^\bigtriangleup
            +
            \widehat{U}(\widehat{U}\widetilde{U}U)^\bigtriangleup
            +
            \widehat{U}(\widehat{U}\widetilde{U}\widehat{\widetilde{U}})^\bigtriangleup
            \nonumber\\
            &+&
            \widehat{U}(\widehat{U}U\widehat{\widetilde{U}})^\bigtriangleup
            +
            \widehat{U}(\widetilde{U}U\widehat{\widetilde{U}})^\bigtriangleup
            +
            \widetilde{U}(\widehat{U}\widetilde{U}U)^\bigtriangleup
            +
            \widetilde{U}(\widehat{U}\widetilde{U}\widehat{\widetilde{U}})^\bigtriangleup
            +
            \widetilde{U}(\widehat{U}U\widehat{\widetilde{U}})^\bigtriangleup
            \nonumber\\
            &+&
            \widetilde{U}(\widetilde{U}U\widehat{\widetilde{U}})^\bigtriangleup
            +
            (\widehat{U}\widetilde{U}U\widehat{\widetilde{U}})^\bigtriangleup 
        ),\label{c4_5}\\
        C_{(5)} &=&
                U\widehat{\widetilde{U}}\widehat{U}\widetilde{U}\widehat{U}^\bigtriangleup
                +
                U\widehat{\widetilde{U}}\widehat{U}\widetilde{U}\widetilde{U}^\bigtriangleup
                +
                U\widehat{\widetilde{U}}\widehat{U}\widetilde{U}U^\bigtriangleup
                +
                U\widehat{\widetilde{U}}\widehat{U}\widetilde{U}\widehat{\widetilde{U}}^\bigtriangleup
                +
                U\widehat{\widetilde{U}}\widehat{U}(\widehat{U}\widetilde{U})^\bigtriangleup
                \nonumber\\
                &+&
                U\widehat{\widetilde{U}}\widehat{U}(\widehat{U}U)^\bigtriangleup
                +
                U\widehat{\widetilde{U}}\widehat{U}(\widehat{U}\widehat{\widetilde{U}})^\bigtriangleup
                +
                U\widehat{\widetilde{U}}\widehat{U}(\widetilde{U}U)^\bigtriangleup
                +
                U\widehat{\widetilde{U}}\widehat{U}(\widetilde{U}\widehat{\widetilde{U}})^\bigtriangleup
                +
                U\widehat{\widetilde{U}}\widehat{U}(U\widehat{\widetilde{U}})^\bigtriangleup
                \nonumber\\
                &+&
                U\widehat{\widetilde{U}}\widetilde{U}(\widehat{U}\widetilde{U})^\bigtriangleup
                +
                U\widehat{\widetilde{U}}\widetilde{U}(\widehat{U}U)^\bigtriangleup
                +
                U\widehat{\widetilde{U}}\widetilde{U}(\widehat{U}\widehat{\widetilde{U}})^\bigtriangleup
                +
                U\widehat{\widetilde{U}}\widetilde{U}(\widetilde{U}U)^\bigtriangleup
                +
                U\widehat{\widetilde{U}}\widetilde{U}(\widetilde{U}\widehat{\widetilde{U}})^\bigtriangleup
                \nonumber\\
                &+&
                U\widehat{\widetilde{U}}\widetilde{U}(U\widehat{\widetilde{U}})^\bigtriangleup
                +
                U\widehat{\widetilde{U}}(\widehat{U}\widetilde{U}U)^\bigtriangleup
                +
                U\widehat{\widetilde{U}}(\widehat{U}\widetilde{U}\widehat{\widetilde{U}})^\bigtriangleup
                +
                U\widehat{\widetilde{U}}(\widehat{U}U\widehat{\widetilde{U}})^\bigtriangleup
                +
                U\widehat{\widetilde{U}}(\widetilde{U}U\widehat{\widetilde{U}})^\bigtriangleup
                \nonumber\\
                &+&
                U\widehat{U}\widetilde{U}(\widehat{U}\widetilde{U})^\bigtriangleup
                +
                U\widehat{U}\widetilde{U}(\widehat{U}U)^\bigtriangleup
                +
                U\widehat{U}\widetilde{U}(\widehat{U}\widehat{\widetilde{U}})^\bigtriangleup
                +
                U\widehat{U}\widetilde{U}(\widetilde{U}U)^\bigtriangleup
                +
                U\widehat{U}\widetilde{U}(\widetilde{U}\widehat{\widetilde{U}})^\bigtriangleup
                \nonumber\\
                &+&
                U\widehat{U}\widetilde{U}(U\widehat{\widetilde{U}})^\bigtriangleup
                +
                U\widehat{U}(\widehat{U}\widetilde{U}U)^\bigtriangleup
                +
                U\widehat{U}(\widehat{U}\widetilde{U}\widehat{\widetilde{U}})^\bigtriangleup
                +
                U\widehat{U}(\widehat{U}U\widehat{\widetilde{U}})^\bigtriangleup 
                +
                U\widehat{U}(\widetilde{U}U\widehat{\widetilde{U}})^\bigtriangleup
                \nonumber\\
                &+&
                U\widetilde{U}(\widehat{U}\widetilde{U}U)^\bigtriangleup
                +
                U\widetilde{U}(\widehat{U}\widetilde{U}\widehat{\widetilde{U}})^\bigtriangleup
                +
                U\widetilde{U}(\widehat{U}U\widehat{\widetilde{U}})^\bigtriangleup
                +
                U\widetilde{U}(\widetilde{U}U\widehat{\widetilde{U}})^\bigtriangleup
                +
                U(\widehat{U}\widetilde{U}U\widehat{\widetilde{U}})^\bigtriangleup
                \nonumber\\
                &+&
                \widehat{\widetilde{U}}\widehat{U}\widetilde{U}(\widehat{U}\widetilde{U})^\bigtriangleup
                +
                \widehat{\widetilde{U}}\widehat{U}\widetilde{U}(\widehat{U}U)^\bigtriangleup
                +
                \widehat{\widetilde{U}}\widehat{U}\widetilde{U}(\widehat{U}\widehat{\widetilde{U}})^\bigtriangleup
                +
                \widehat{\widetilde{U}}\widehat{U}\widetilde{U}(\widetilde{U}U)^\bigtriangleup
                +
                \widehat{\widetilde{U}}\widehat{U}\widetilde{U}(\widetilde{U}\widehat{\widetilde{U}})^\bigtriangleup
                \nonumber\\
                &+&
                \widehat{\widetilde{U}}\widehat{U}\widetilde{U}(U\widehat{\widetilde{U}})^\bigtriangleup
                +
                \widehat{\widetilde{U}}\widehat{U}(\widehat{U}\widetilde{U}U)^\bigtriangleup
                +
                \widehat{\widetilde{U}}\widehat{U}(\widehat{U}\widetilde{U}\widehat{\widetilde{U}})^\bigtriangleup
                +
                \widehat{\widetilde{U}}\widehat{U}(\widehat{U}U\widehat{\widetilde{U}})^\bigtriangleup
                +
                \widehat{\widetilde{U}}\widehat{U}(\widetilde{U}U\widehat{\widetilde{U}})^\bigtriangleup
                \nonumber\\
                &+&
                \widehat{\widetilde{U}}\widetilde{U}(\widehat{U}\widetilde{U}U)^\bigtriangleup
                +
                \widehat{\widetilde{U}}\widetilde{U}(\widehat{U}\widetilde{U}\widehat{\widetilde{U}})^\bigtriangleup
                +
                \widehat{\widetilde{U}}\widetilde{U}(\widehat{U}U\widehat{\widetilde{U}})^\bigtriangleup
                +
                \widehat{\widetilde{U}}\widetilde{U}(\widetilde{U}U\widehat{\widetilde{U}})^\bigtriangleup
                +
                \widehat{\widetilde{U}}(\widehat{U}\widetilde{U}U\widehat{\widetilde{U}})^\bigtriangleup
                \nonumber\\
                &+&
                \widehat{U}\widetilde{U}(\widehat{U}\widetilde{U}U)^\bigtriangleup
                +
                \widehat{U}\widetilde{U}(\widehat{U}\widetilde{U}\widehat{\widetilde{U}})^\bigtriangleup
                +
                \widehat{U}\widetilde{U}(\widehat{U}U\widehat{\widetilde{U}})^\bigtriangleup
                +
                \widehat{U}\widetilde{U}(\widetilde{U}U\widehat{\widetilde{U}})^\bigtriangleup
                +
                \widehat{U}(\widehat{U}\widetilde{U}U\widehat{\widetilde{U}})^\bigtriangleup
                \nonumber\\
                &+&
                \widetilde{U}(\widehat{U}\widetilde{U}U\widehat{\widetilde{U}})^\bigtriangleup
            ,\label{c5_5}\\
    C_{(6)} &=&
    -(
        U\widehat{\widetilde{U}}\widehat{U}\widetilde{U}(\widehat{U}\widetilde{U})^\bigtriangleup
        +
        U\widehat{\widetilde{U}}\widehat{U}\widetilde{U}(\widehat{U}U)^\bigtriangleup
        +
        U\widehat{\widetilde{U}}\widehat{U}\widetilde{U}(\widehat{U}\widehat{\widetilde{U}})^\bigtriangleup
        +
        U\widehat{\widetilde{U}}\widehat{U}\widetilde{U}(\widetilde{U}U)^\bigtriangleup
        \nonumber\\
        &+&
        U\widehat{\widetilde{U}}\widehat{U}\widetilde{U}(\widetilde{U}\widehat{\widetilde{U}})^\bigtriangleup
        +
        U\widehat{\widetilde{U}}\widehat{U}\widetilde{U}(U\widehat{\widetilde{U}})^\bigtriangleup
        +
        U\widehat{\widetilde{U}}\widehat{U}(\widehat{U}\widetilde{U}U)^\bigtriangleup
        +
        U\widehat{\widetilde{U}}\widehat{U}(\widehat{U}\widetilde{U}\widehat{\widetilde{U}})^\bigtriangleup
        \nonumber\\
        &+&
        U\widehat{\widetilde{U}}\widehat{U}(\widehat{U}U\widehat{\widetilde{U}})^\bigtriangleup
        +
        U\widehat{\widetilde{U}}\widehat{U}(\widetilde{U}U\widehat{\widetilde{U}})^\bigtriangleup
        +
        U\widehat{\widetilde{U}}\widetilde{U}(\widehat{U}\widetilde{U}U)^\bigtriangleup
        +
        U\widehat{\widetilde{U}}\widetilde{U}(\widehat{U}\widetilde{U}\widehat{\widetilde{U}})^\bigtriangleup
        \nonumber\\
        &+&
        U\widehat{\widetilde{U}}\widetilde{U}(\widehat{U}U\widehat{\widetilde{U}})^\bigtriangleup
        +
        U\widehat{\widetilde{U}}\widetilde{U}(\widetilde{U}U\widehat{\widetilde{U}})^\bigtriangleup
        +
        U\widehat{\widetilde{U}}(\widehat{U}\widetilde{U}U\widehat{\widetilde{U}})^\bigtriangleup
        +
        U\widehat{U}\widetilde{U}(\widehat{U}\widetilde{U}U)^\bigtriangleup
        \nonumber\\
        &+&
        U\widehat{U}\widetilde{U}(\widehat{U}\widetilde{U}\widehat{\widetilde{U}})^\bigtriangleup
        +
        U\widehat{U}\widetilde{U}(\widehat{U}U\widehat{\widetilde{U}})^\bigtriangleup
        +
        U\widehat{U}\widetilde{U}(\widetilde{U}U\widehat{\widetilde{U}})^\bigtriangleup
        +
        U\widehat{U}(\widehat{U}\widetilde{U}U\widehat{\widetilde{U}})^\bigtriangleup
        \nonumber\\
        &+&
        U\widetilde{U}(\widehat{U}\widetilde{U}U\widehat{\widetilde{U}})^\bigtriangleup
        +
        \widehat{\widetilde{U}}\widehat{U}\widetilde{U}(\widehat{U}\widetilde{U}U)^\bigtriangleup
        +
        \widehat{\widetilde{U}}\widehat{U}\widetilde{U}(\widehat{U}\widetilde{U}\widehat{\widetilde{U}})^\bigtriangleup
        +
        \widehat{\widetilde{U}}\widehat{U}\widetilde{U}(\widehat{U}U\widehat{\widetilde{U}})^\bigtriangleup
        \nonumber\\
        &+&
        \widehat{\widetilde{U}}\widehat{U}\widetilde{U}(\widetilde{U}U\widehat{\widetilde{U}})^\bigtriangleup
        +
        \widehat{\widetilde{U}}\widehat{U}(\widehat{U}\widetilde{U}U\widehat{\widetilde{U}})^\bigtriangleup
        +
        \widehat{\widetilde{U}}\widetilde{U}(\widehat{U}\widetilde{U}U\widehat{\widetilde{U}})^\bigtriangleup
        +
        \widehat{U}\widetilde{U}(\widehat{U}\widetilde{U}U\widehat{\widetilde{U}})^\bigtriangleup
    )
    ,\label{c6_5}\\
    C_{(7)} &=&
    U\widehat{\widetilde{U}}\widehat{U}\widetilde{U}(\widehat{U}\widetilde{U}U)^\bigtriangleup
    +
    U\widehat{\widetilde{U}}\widehat{U}\widetilde{U}(\widehat{U}\widetilde{U}\widehat{\widetilde{U}})^\bigtriangleup
    +
    U\widehat{\widetilde{U}}\widehat{U}\widetilde{U}(\widehat{U}U\widehat{\widetilde{U}})^\bigtriangleup
    \nonumber\\
    &+&
    U\widehat{\widetilde{U}}\widehat{U}\widetilde{U}(\widetilde{U}U\widehat{\widetilde{U}})^\bigtriangleup
    +
    U\widehat{\widetilde{U}}\widehat{U}(\widehat{U}\widetilde{U}U\widehat{\widetilde{U}})^\bigtriangleup
    +
    U\widehat{\widetilde{U}}\widetilde{U}(\widehat{U}\widetilde{U}U\widehat{\widetilde{U}})^\bigtriangleup
    \nonumber\\
    &+&
    U\widehat{U}\widetilde{U}(\widehat{U}\widetilde{U}U\widehat{\widetilde{U}})^\bigtriangleup
    +
    \widehat{\widetilde{U}}\widehat{U}\widetilde{U}(\widehat{U}\widetilde{U}U\widehat{\widetilde{U}})^\bigtriangleup
    , \label{c7_5}\\
    C_{(8)} &=&
    -
    U\widehat{\widetilde{U}}\widehat{U}\widetilde{U}(\widehat{U}\widetilde{U}U\widehat{\widetilde{U}})^\bigtriangleup
    . \label{c8_5}
    \end{eqnarray}
}
\end{thm}
\begin{proof} 
We verified the basis-free formulas (\ref{c1_5}) -- (\ref{c8_5}) for $C_1$, \ldots, $C_8$ using Symbolic Geometric Algebra package for SymPy \cite{Python2}.

We also present an analytical proof of the basis-free formula (\ref{c2_5}) for $C_{(2)}$.

Step 1: Let us denote
    \begin{eqnarray}
        d_1 :=
            U\widehat{\widetilde{U}}  
            + 
             \widehat{U}\widetilde{U}  
            +
             (\widehat{U}\widetilde{U})^\bigtriangleup  
            +
             (U\widehat{\widetilde{U}})^\bigtriangleup,\nonumber \\
        d_2 :=  
            U\widehat{U}  
            +
            \widehat{\widetilde{U}}\widetilde{U} 
            +
             (\widehat{U}U)^\bigtriangleup 
            +
             (\widetilde{U}\widehat{\widetilde{U}})^\bigtriangleup, \nonumber \\
        d_3 :=
            U\widetilde{U}  
            +
            \widehat{\widetilde{U}}\widehat{U}
            +
            (\widehat{U}\widehat{\widetilde{U}})^\bigtriangleup
            +
            (\widetilde{U}U)^\bigtriangleup,\nonumber \\
        d_4 :=
            U\widehat{U}^\bigtriangleup  
            +
            \widehat{\widetilde{U}}\widetilde{U}^\bigtriangleup
            +
            \widetilde{U}\widehat{\widetilde{U}}^\bigtriangleup
            +
            \widehat{U}U^\bigtriangleup,\nonumber\\
        d_5 :=
            U\widetilde{U}^\bigtriangleup  
            +
            \widehat{\widetilde{U}}\widehat{U}^\bigtriangleup
            +
             \widetilde{U}U^\bigtriangleup
            +
             \widehat{U}\widehat{\widetilde{U}}^\bigtriangleup,\nonumber \\
        d_6 :=
            UU^\bigtriangleup  
            +
             \widehat{U}\widehat{U}^\bigtriangleup 
            +
             \widetilde{U}\widetilde{U}^\bigtriangleup 
            +
             \widehat{\widetilde{U}}\widehat{\widetilde{U}}^\bigtriangleup\nonumber,\\
        d_7 :=
            U\widehat{\widetilde{U}}^\bigtriangleup
            +
             \widehat{\widehat{U}}U^\bigtriangleup
            +
             \widehat{U}\widetilde{U}^\bigtriangleup 
            +
             \widetilde{U}\widehat{U}^\bigtriangleup .\label{proof_c3_5_c1}
    \end{eqnarray}
    Using the properties (\ref{conj_prop2}) and (\ref{conj_prop3}), we can prove that the projections onto the subspace of the grade $0$ of all terms in each of the 7 expressions (\ref{proof_c3_5_c1}) are equal to each other. For example,  
    $$
            \frac{1}{4}\langle d_1 \rangle_0=
            \langle U\widehat{\widetilde{U}} \rangle_0 
            = 
            \langle \widehat{U}\widetilde{U} \rangle_0 
            =
            \langle (\widehat{U}\widetilde{U})^\bigtriangleup \rangle_0 
            =
            \langle (U\widehat{\widetilde{U}})^\bigtriangleup \rangle_0
    .$$
    Using (\ref{recur}), (\ref{proj}), (\ref{c1_5}), and (\ref{proof_c3_5_c1}), we get
    \begin{eqnarray}
           \!\!\!\!\!\! &&C_{(2)}=4\langle U(U-C_{(1)})\rangle_0
            =-4\langle U(
                \widehat{\widetilde{U}}
                +
                \widehat{U}
                +
                \widetilde{U}
                +
                \widehat{U}^\bigtriangleup
                +
                \widetilde{U}^\bigtriangleup
                +
                U^\bigtriangleup
                +
                \widehat{\widetilde{U}}^\bigtriangleup
            )\rangle_0\nonumber\\
           \!\!\!\!\!\!&&=-\langle d_1+d_2+d_3+d_4+d_5+d_6+d_7\rangle_0\nonumber
            ,
    \end{eqnarray}
    which differs from (\ref{c2_5}) only by the scalar part operation.
    
    Step 2:
    It follows from (\ref{u+u^g}), (\ref{u+u^r}), and $\widetilde{d_3}=\widehat{d_3}^{\bigtriangleup}=d_3$ that $d_1\in\cl_{p,q}^0$, $d_2, d_3\in\cl_{p,q}^{0}\oplus\cl_{p,q}^5$. However, using (\ref{conjug}) and $\langle UV\rangle_n = \langle VU\rangle_n$ for odd $n$, we can prove that  $d_2$ and $d_3$  belong to $
    \cl_{p,q}^0$:
    \begin{eqnarray}
        \langle d_2 \rangle_5
        =
        \langle
        U\widehat{U}  
        +
        \widehat{\widetilde{U}}\widetilde{U} 
        -
         \widehat{U}U
        -
         \widetilde{U}\widehat{\widetilde{U}} 
        \rangle_5 = 
        \langle
        U\widehat{U}  
        +
        \widehat{\widetilde{U}}\widetilde{U} 
        -
         U\widehat{U}
        -
         \widehat{\widetilde{U}}\widetilde{U}
        \rangle_5
        =0
        ,\nonumber\\
        \langle d_3 \rangle_5
        =
        \langle
        U\widetilde{U}  
        +
        \widehat{\widetilde{U}}\widehat{U}
        -
        \widehat{U}\widehat{\widetilde{U}}
        -
        \widetilde{U}U
        \rangle_5
        =
        \langle
        U\widetilde{U}  
        +
        \widehat{\widetilde{U}}\widehat{U}
        -
        \widehat{\widetilde{U}}\widehat{U}
        -
        U\widetilde{U}
        \rangle_5
        =0
        .\nonumber
    \end{eqnarray}
    Using (\ref{u+u^gg}), we get that the following sum can be represented in the form
    \begin{eqnarray}
        &&d_4+d_5+d_6+d_7=
        (U+\widehat{U}+\widetilde{U}+\widehat{\widetilde{U}})
        (U+\widehat{U}+\widetilde{U}+\widehat{\widetilde{U}})^\bigtriangleup\nonumber\\
        &&=16(
        U_0
        +
        U_4)
        (
        U_0
        -
        U_4)
        =
        16(
        U_0^2
        -
        U_4^2
        )\in\cl^0_{p,q}
        ,
        \nonumber
        \label{proof_c2_5_c4-c7}
    \end{eqnarray}
because $U_0^2, U_4^2 \in\cl_{p,q}^0$. Therefore the expression (\ref{c2_5}) belongs to $\cl_{p,q}^0$.
\end{proof}

\section{Formulas for the Characteristic Polynomial Coefficients in the Special Cases}\label{section6} 
In this section, we present the basis-free formulas for the characteristic polynomial coefficients in the cases with some specific conditions on the element. In particular, we obtain the basis-free formulas for vectors and basis elements in the case of arbitrary $n$, and rotors in the cases $n\leq 5$. 

We remind that $N:= 2^{[\frac{n+1}{2}]}$, where square brackets mean taking the integer part.
\begin{lem}\label{thm:ch_poly_e}
    For $U=ue$, $u\in\mathbb{R}$, where $e$ is the identity element of $\cl_{p,q}$, we have
    \begin{eqnarray}
    C_{(i)}=(-1)^{i+1}\binom{N}{i}U^{i}, \qquad i=1,\ldots,N,\label{q2}
    \end{eqnarray}
    where we denote the binomial coefficient by $\binom{n}{k}=\frac{n!}{k!(n-k)!}$.
\end{lem}
\begin{proof}
    Using Definition \ref{def:ch_poly}, we get
    \begin{eqnarray}
        \varphi_U(\lambda) &=&{\rm Det}(\lambda e - ue)={\rm Det}((\lambda - u)e)\nonumber\\
        &=&(\lambda - u)^N{\rm Det}(e)=\sum_{i=0}^{N} \binom{N}{i}(-u)^i\lambda^{N-i}.\label{q1}
    \end{eqnarray}
    Comparing (\ref{q1}) with (\ref{ch_poly}),
    we obtain (\ref{q2}).
\end{proof}
\begin{lem}\label{lemma:uk_show}
The elements $U_{(k)}$ and the characteristic polynomial coefficients $C_{(k)}$ from Theorem \ref{theorem:mainTheorem} could be expressed in the following way:
\begin{eqnarray}
    U_{(k)}&=& U^k-\sum_{i=1}^{k-1}C_{(i)} U^{k-i}, \quad 
    k=1,2,\ldots,N, \label{lemm:uk_show}\\
    C_{(k)}&=&\frac{N}{k}\langle U^k\rangle_0-\frac{N}{k}\sum_{i=1}^{k-1}C_{(i)}\langle U^{k-i}\rangle_0, \quad
    k=1,2,\ldots,N.\label{lemm:ck_show}
\end{eqnarray}
\end{lem}
\begin{proof}
    Using the recursive formula for $U_{(k)}$ (\ref{recur})  $k$ times, we obtain (\ref{lemm:uk_show}). Combining the formula for $C_{(k)}$ (\ref{recur}) and (\ref{lemm:uk_show}), we obtain (\ref{lemm:ck_show}).
\end{proof}
\begin{thm}\label{lemma:c_2s-1=0}
    Let us have an element $U\in\cl_{p,q}$  and some $k\in\{1,2,\ldots,\frac{N}{2}\}$ such that
    \begin{eqnarray}
        \langle U^{2s-1} \rangle_0 =0,\qquad  s=1, 2, \ldots, k.\label{proof:induc_odd}
    \end{eqnarray} 
    Then the characteristic polynomial coefficients with odd indices of the element $U$ are equal to~$0$:
    \begin{eqnarray}
        C_{(2s-1)}=0,\qquad s=1, 2, \ldots, k.\label{lem:c_2s+1=0}
    \end{eqnarray}
    \end{thm}
    \begin{proof}
        The proof is by induction on $s$.
        First let us check the base case $s=1$. Using (\ref{recur}) and (\ref{proof:induc_odd}), we get
        $$C_{(1)}=N\langle U_{(1)} \rangle_0=N\langle U \rangle_0=0.$$
         For the inductive step, assume that for some $k\in\{1,2,\ldots,\frac{N}{2}-1\}$ the following holds 
        \begin{eqnarray}
             C_{(2s-1)}=0, \quad s=1,2,\ldots,k. \label{proof:induc_hypo}
        \end{eqnarray}
        Now let us show that (\ref{proof:induc_hypo}) holds for $s=k+1$. Using Lemma \ref{lemma:uk_show}, we obtain
        \begin{eqnarray}
            C_{(2k+1)}=\frac{N}{2k+1}\langle U^{2k+1}\rangle_0-\frac{N}{2k+1}\sum_{i=1}^{2k}C_{(i)}\langle U^{2k+1-i}\rangle_0.\label{proof:induc_exp1}
        \end{eqnarray}
        By the induction hypothesis (\ref{proof:induc_hypo}), the expression (\ref{proof:induc_exp1}) could be simplified as follows
        \begin{eqnarray}
            C_{(2k+1)}=\frac{N}{2k+1}\langle U^{2k+1}\rangle_0-\frac{N}{2k+1}\sum_{j=1}^{k}C_{(2j)}\langle U^{2k+1-2j}\rangle_0.\nonumber
        \end{eqnarray}
    Using (\ref{proof:induc_odd}) and the fact that $2k+1$ and $2k+1-2j$ are odd for all $k$ and $j$, we get $C_{(2k+1)}=0$. This concludes the induction step and the proof.
\end{proof}
\begin{thm}\label{thm:like_basis_el}
    Let us have an element $U\in\cl_{p,q}$ such that 
    \begin{eqnarray}
        \langle U \rangle_0 =0, \qquad
        U^{2}\in \cl_{p,q}^0.\label{lem:el_similar_tobasis}
    \end{eqnarray}
    Then the following formulas hold
    \begin{eqnarray}
        &&C_{(2s-1)}=0,\qquad s=1,2,\ldots,\frac{N}{2},\label{lem:c2k-1=0}\\
        &&C_{(2s)}=(-1)^{s-1}\binom{\frac{N}{2}}{s}U^{2s}, \qquad s=1,2,\ldots,\frac{N}{2}.\label{lem:c2k_show}
    \end{eqnarray}
\end{thm}
\begin{proof}
It follows from (\ref{lem:el_similar_tobasis}) that 
    \begin{eqnarray}
        \langle U^{2s-1} \rangle_0 =0, \quad
        U^{2s}\in \cl_{p,q}^0,\quad
        s=1,2,\ldots,\frac{N}{2}.\label{u^2=num=>}
    \end{eqnarray}
Hence using Theorem \ref{lemma:c_2s-1=0}, we get (\ref{lem:c2k-1=0}).

The proof of (\ref{lem:c2k_show}) is by induction on $s$. First we need to prove some auxiliary identities and statements. Using (\ref{lemm:uk_show}), we get
\begin{eqnarray}
    U_{(2s)}=U^{2s}-\sum_{i=1}^{2s-1}C_{(i)} U^{2s-i}, \quad s=1,2,\ldots,\frac{N}{2}.\label{exp_th-num}
\end{eqnarray}
Using (\ref{lem:c2k-1=0}), we obtain the following simplification of (\ref{exp_th-num}):
\begin{eqnarray}
    U_{(2s)}= U^{2s}-\sum_{j=1}^{s}C_{(2j)} U^{2s-2j}.\nonumber
\end{eqnarray}
Using (\ref{u^2=num=>}) and the fact that $2s$ and $2s-2j$ are even for all $s$ and $j$, we finally get
\begin{eqnarray}
    U_{(2s)}\in \cl_{p,q}^0.\label{lem:u2k_isgrade0}
\end{eqnarray}
Using (\ref{recur}) and (\ref{lem:u2k_isgrade0}), we get
\begin{eqnarray}
     C_{(2s)}=\frac{N}{2s}\langle U_{(2s)}\rangle_0=\frac{N}{2s}U_{(2s)}.\label{lem:proof_c2k_u2k}
\end{eqnarray}
Using the recursive formulas (\ref{recur}) two times, (\ref{lem:c2k-1=0}), and (\ref{lem:proof_c2k_u2k}), we obtain 
    \begin{eqnarray}
     C_{(2s)}&=&\frac{N}{2s}U_{(2s)}=\frac{N}{2s}U(U_{(2s-1)}-C_{(2s-1)})=\frac{N}{2s}U U_{(2s-1)}\label{lem:proof_c2k_u2k_2}\\
     &=&\frac{N}{2s}U^2(U_{(2s-2)}-C_{(2s-2)}).\nonumber
    \end{eqnarray}
 Using (\ref{lem:proof_c2k_u2k}), we finally get the following recursive identity
    \begin{eqnarray}
     C_{(2s)}
     &=&\frac{N}{2s}U^2(\frac{2s-2}{N}C_{(2s-2)}-C_{(2s-2)})\nonumber\\
     &=&\frac{N}{2s}U^2(\frac{2s-2-N}{N})C_{(2s-2)}=\frac{2s-2-N}{2s}C_{(2s-2)}U^2.\label{lem:c2k_c2k}
    \end{eqnarray}
    
Now let us proceed to the actual proof of (\ref{lem:c2k_show}) by induction on $s$. First let us check the base case $s = 1$. The base case follows from (\ref{lem:proof_c2k_u2k_2}): $$C_{(2)}=\frac{N}{2}U^2=\binom{\frac{N}{2}}{1}U^2.$$
For the inductive step, assume that for some $k\in\{1,2,\ldots,\frac{N}{2}-1\}$
    \begin{eqnarray}
        C_{(2s)}=(-1)^{s-1}\binom{\frac{N}{2}}{s}U^{2s}, \qquad 
        s=1,2,\ldots,k.\label{proof:ind_step1}
    \end{eqnarray}
Now let us show that (\ref{proof:ind_step1}) holds for $s=k + 1$. Using the formula (\ref{lem:c2k_c2k}), the inductive hypothesis (\ref{proof:ind_step1}), and the identity $\binom{n}{k}=\frac{n-k+1}{k}\binom{n}{k-1}$, we get
\begin{eqnarray}
    &&C_{(2k+2)}=\frac{k-\frac{N}{2}}{k+1} C_{(2k)} U^{2}=\frac{k-\frac{N}{2}}{k+1}(-1)^{k-1}\binom{\frac{N}{2}}{k}U^{2k}U^2\nonumber\\
    &&=(-1)^{k}\frac{\frac{N}{2}-(k+1)+1}{k+1}\binom{\frac{N}{2}}{(k+1)-1}U^{2k}U^2=(-1)^{k}\binom{\frac{N}{2}}{k+1}U^{2k+2},\nonumber
\end{eqnarray}which concludes the induction step and the proof.
\end{proof}
\begin{ex}
For elements $U\in\cl_{p,q}$ that satisfy the conditions of Theorem~\ref{thm:like_basis_el}, we have
\begin{eqnarray}
n=1,2, 
&&C_{(1)}=0,\quad
    C_{(2)}=U^2;\nonumber\\
n=3,4,  
&&C_{(1)}=C_{(3)}=0,\quad
    C_{(2)}=2U^2,\quad
    C_{(4)}=-U^4;\nonumber\\
n=5,6,  &&C_{(1)}=C_{(3)}=C_{(5)}=C_{(7)}=0,\nonumber\\
&&C_{(2)}=4U^2,\quad
    C_{(4)}=-6U^4,\quad
    C_{(6)}=4U^6,\quad
    C_{(8)}=-U^8;\nonumber\\
n=7,8,  &&C_{(1)}=C_{(3)}=C_{(5)}=C_{(7)}=C_{(9)}=C_{(11)}=C_{(13)}=C_{(15)}=0,\nonumber\\
    &&C_{(2)}=8U^2,\quad
    C_{(4)}=-28U^4,\quad
    C_{(6)}=56U^6,\quad
    C_{(8)}=-70U^8,\nonumber\\
    &&C_{(10)}=56U^{10},\,\,
    C_{(12)}=-28U^{12},\,\,
    C_{(14)}=8U^{14},\,\,
    C_{(16)}=-U^{16}.\nonumber
 \end{eqnarray}
\end{ex}
\begin{ex}
Let us consider an element of the following form
        \begin{eqnarray}
            U=ue_{a_1 \ldots a_k},\qquad  u\in\mathbb{R},  \nonumber
        \end{eqnarray}
        where 
        $e_{a_1 \ldots a_k}$ is an basis element from (\ref{def:basis}) except the identity element $e$. The element $U$ satisfies the conditions of Theorem \ref{thm:like_basis_el}. Thus the formulas (\ref{lem:c2k-1=0}) and (\ref{lem:c2k_show}) are valid for the element $U$.
\end{ex}
\begin{ex}
    An element of grade $1$ (a vector) 
    $$U=\sum_{i=1}^{n}u_ie_i\in \cl_{p,q}^1,\qquad u_i\in\mathbb{R},$$  
    satisfies the conditions of Theorem \ref{thm:like_basis_el}.  Thus the formulas (\ref{lem:c2k-1=0}) and (\ref{lem:c2k_show}) are valid for the vector $U$.
\end{ex}
\begin{ex}
    Let us consider the following element
    $$U=e_1+e_5+e_{15}\in \cl_{5,0}.$$
    For this element, we have
    \begin{eqnarray}
        \langle U\rangle_0 =0,\qquad 
        U^{2}=e,\label{ex:prop_e1+e2+e12}
    \end{eqnarray}
    meaning the element $U$ satisfies the conditions of Theorem \ref{thm:like_basis_el}. Using Theorem \ref{thm:like_basis_el}, we get
    \begin{eqnarray*}
    &&C_{(1)}=C_{(3)}=C_{(5)}=C_{(7)}=0,\\
    &&C_{(2)}=4,\quad C_{(4)}=-6,\quad C_{(6)}=4,\quad C_{(8)}=-1.
    \end{eqnarray*}
\end{ex}
\begin{ex}
    Let us consider the following element  
    $$U=e_1+e_2+e_{45}\in \cl_{5,0},\qquad N=8.$$
     For this element, we have
    \begin{eqnarray}
        U^{2}\in\cl_{p,q}^0\oplus\cl_{p,q}^3,\qquad  
        U^{2}\notin \cl_{p,q}^0,\qquad
        \langle U^{2s-1}\rangle_0 =0,\quad
        s=1, 2, 3, 4,\label{ex:prop_e1+e2+e45}
    \end{eqnarray}
    meaning the element $U$ satisfies the conditions of Theorem \ref{lemma:c_2s-1=0}, but does not satisfy the conditions of Theorem~\ref{thm:like_basis_el}. Using Theorem \ref{lemma:c_2s-1=0}, we get
    $$C_{(1)}=C_{(3)}=C_{(5)}=C_{(7)}=0.$$
    Using the recursive formulas from Theorem \ref{theorem:mainTheorem}, we verified  that $C_{(2)}=4$, which is not equal to $4U^2$. 
\end{ex}
\begin{ex}\label{ex_th63_2}
    Let us consider the following element
$$U=e_{3}+e_{12}+e_{15}+e_{45}+e_{234}\in \cl_{5,0}, \qquad N=8$$
    For this element, we have 
    \begin{eqnarray}
        \langle U\rangle_0 = \langle U^3\rangle_0 = 0,\qquad 
        \langle U^{5}\rangle_0 \neq 0,\nonumber
    \end{eqnarray}
   meaning the element $U$ satisfies the conditions of Theorem \ref{lemma:c_2s-1=0} for $k=2$, which is less than $\frac{N}{2}=4$. Using Theorem \ref{lemma:c_2s-1=0}, we get
    $$C_{(1)}=C_{(3)}=0.$$
    Using the recursive formulas from Theorem \ref{theorem:mainTheorem}, we verified  that $C_{(5)}=-64\neq0$. This means that the condition on $k$ in Theorem \ref{lemma:c_2s-1=0} is significant.
\end{ex}
\bigskip
Let us consider the following well-known spin groups \cite{Lounesto, Hestenes, Doran}
\begin{eqnarray}
    {\rm Spin_{+}}(p,q)=\{ 
    U\in \cl_{p,q} \,\,|\,\,  \widehat{U}=U,\quad  \widetilde{U}U=e,\quad 
       U\cl_{p,q}^1U^{-1}\subset \cl_{p,q}^1\}.\label{def:spin} 
\end{eqnarray}
The elements of these groups are often called rotors. Spin groups have a wide range of applications in physics, computer science, and engineering. In the following theorem, we present basis-free formulas for characteristic polynomial coefficients of rotors in the cases $n\leq 5$.  
\begin{thm}\label{thm:spin} For an arbitrary element $U\in {\rm Spin}_{+}(p,q)$, we have the following basis-free formulas for the characteristic polynomial coefficients $C_{(k)}\in\cl^0_{p,q}$, $k=1, 2, \ldots, N$.\\
In the case $n=1$, we have
\begin{eqnarray*}
    C_{(1)}=2U,\qquad
    C_{(2)}=-e.
\end{eqnarray*}
In the case $n=2$, we have
\begin{eqnarray*}
    C_{(1)}=U+\widetilde{U},\qquad  C_{(2)}=-e.
\end{eqnarray*}
In the case $n=3$, we have
\begin{eqnarray*}
    &&C_{(1)}=C_{(3)}=
    2(U+\widetilde{U}),\quad
    C_{(2)}=-(
    4e+
    U^2+
    \widetilde{U}^2
    ),
    \quad C_{(4)}=-e.
\end{eqnarray*}
In the case $n=4$, we have
\begin{eqnarray*}
    C_{(1)}=C_{(3)}&=&
    U+\widetilde{U}+U^\bigtriangleup+\widetilde{U}^\bigtriangleup,\nonumber\\
    C_{(2)}&=&-(
    2e+
    UU^\bigtriangleup+
    U\widetilde{U}^\bigtriangleup+\widetilde{U}U^\bigtriangleup+
    \widetilde{U}\widetilde{U}^\bigtriangleup),\nonumber\\
    C_{(4)}&=&-e.
\end{eqnarray*}
In the case $n=5$, we have
\begin{eqnarray*}
    C_{(1)}=C_{(7)}&=&
    2(U+\widetilde{U}+U^\bigtriangleup+\widetilde{U}^\bigtriangleup),\\
    C_{(2)}=C_{(6)}&=&-(
    8e+
    U^2+
    \widetilde{U}^2+
    (U^2)^\bigtriangleup+
    (\widetilde{U}^2)^\bigtriangleup
    +
    4(UU^\bigtriangleup+
    U\widetilde{U}^\bigtriangleup\\
    &&+
    \widetilde{U}U^\bigtriangleup+
    \widetilde{U}\widetilde{U}^\bigtriangleup)),\\
    C_{(3)}=C_{(5)}&=&
    10(U+\widetilde{U}+U^\bigtriangleup+\widetilde{U}^\bigtriangleup)+2(
        U^2U^\bigtriangleup+
        U^2\widetilde{U}^\bigtriangleup+
        \widetilde{U}^2U^\bigtriangleup\\
        &&+
        \widetilde{U}^2\widetilde{U}^\bigtriangleup
        +U(U^2)^\bigtriangleup+
        U(\widetilde{U}^2)^\bigtriangleup+
        \widetilde{U}(U^2)^\bigtriangleup+
        \widetilde{U}(\widetilde{U}^2)^\bigtriangleup),\\\
    C_{(4)}&=&-(18e+8(
    UU^\bigtriangleup+
    U\widetilde{U}^\bigtriangleup+
    \widetilde{U}U^\bigtriangleup+
    \widetilde{U}\widetilde{U}^\bigtriangleup)
    +4(
    (U^2)^\bigtriangleup\\
    &&+
    (\widetilde{U}^2)^\bigtriangleup
    +
    U^2+
    \widetilde{U}^2)
    +U^2(U^2)^\bigtriangleup
    +U^2(\widetilde{U}^2)^\bigtriangleup
    +\widetilde{U}^2(U^2)^\bigtriangleup\\
    &&+
    \widetilde{U}^2(\widetilde{U}^2)^\bigtriangleup),\\
    C_{(8)}&=&-e.
\end{eqnarray*}
\end{thm}
\begin{proof}
    The formulas follow straightforward from Theorem \ref{theorem2},  Theorem \ref{thm:ck_n=5} and the conditions $\widehat{U}=U$ and $\widetilde{U}U=e$ from the definition of the spin group~(\ref{def:spin}).
\end{proof}
 Note that similarly we can obtain the basis-free formulas for the characteristic polynomial coefficients for an arbitrary element $U\in {\rm Spin}_{+}(p,q)$ in the case $n=p+q=6$ by simplification the basis-free formulas from Section~\ref{section8} and Appendix~\ref{appendix} and using the conditions $\widehat{U}=U$ and $\widetilde{U}U=e$ from the definition of the spin group~(\ref{def:spin}). We do not present these formulas in this paper because of their cumbersomeness. 

\section{The General Form of the Obtained Formulas for the Characteristic Polynomial Coefficients}\label{section7}
In this section, we introduce a method to obtain a general form of the basis-free formulas for all characteristic polynomial coefficients in the cases $n\leq 5$. The essence of the method is illustrated by the following simple example.
\begin{ex} 
    Using (\ref{c4_4}) and (\ref{main_det}) for $n=4$, we get
    $${\rm Det}(U)=U\widehat{\widetilde{U}}(\widehat{U}\widetilde{U})^\bigtriangleup\in \cl_{p,q},\qquad U\in \cl_{p,q}.$$
    By enumerating every occurrence of the element $U$ from left to right in the basis-free formula ${\rm Det}(U)$, we get the following multi-variable function 
    \begin{eqnarray}
    {\rm F}(x_1,x_2,x_3,x_4)=x_1\widehat{\widetilde{x_2}}(\widehat{x_3}\widetilde{x_4})^\bigtriangleup\in \cl_{p,q},\qquad x_{i}\in\cl_{p,q},\quad
    i=1,2,3,4.\label{ex:f-det4}
    \end{eqnarray}
    In the case $n=4$, all the characteristic polynomial coefficients $C_{(1)}$, $C_{(2)}$, $C_{(3)}$, $C_{(4)}$ could be expressed using (\ref{ex:f-det4}) as follows 
    \begin{eqnarray}
        C_{(1)}&=&
            {\rm F}(U,e,e,e)+
            {\rm F}(e,U,e,e)+
            {\rm F}(e,e,U,e)+
            {\rm F}(e,e,e,U)\nonumber\\
        C_{(2)}&=&-(
            {\rm F}(U,U,e,e)+
            {\rm F}(U,e,U,e)+
            {\rm F}(U,e,e,U)+
            {\rm F}(e,U,U,e)\nonumber\\
            &&+
            {\rm F}(e,U,e,U)+
            {\rm F}(e,e,U,U))\nonumber\\
        C_{(3)}&=&
            {\rm F}(e,U,U,U)+
            {\rm F}(U,e,U,U)+
            {\rm F}(U,U,e,U)+
            {\rm F}(U,U,U,e)\nonumber\\
        C_{(4)}&=&
            -{\rm F}(U,U,U,U).\nonumber
    \end{eqnarray}
    Therefore one could rewrite the formulas for the characteristic polynomial coefficients in the following general form
    $$C_{(k)}=(-1)^{k+1}\sum_{X_j \in X(k)}
            {\rm F}(X_j),\qquad
            k=1,2,3,4,$$
    where $X(k)$ is the set of all possible tuples with $k$ elements $U$ and $4-k$ identity elements $e$.
\end{ex}
Using similar method, we obtain the general form of the basis-free formulas for all characteristic polynomial coefficients for the cases $n\leq 5$ in the following theorem.
\begin{thm}\label{thm:cross5}
    In the cases $n\leq 5$, the basis-free formulas for the characteristic polynomial coefficients $C_{(k)}\in\cl_{p,q}$, $k = 1,2,\ldots,N$ have the following form
    \begin{eqnarray}
        C_{(k)}=(-1)^{k+1}\sum_{X_j\in X(k)}{\rm F}(X_j),\quad
        k = 1,2,\ldots,N,\label{cross<5}
    \end{eqnarray}
    where ${\rm F}$ is defined as follows
    \begin{eqnarray}
        n=1,&&
        {\rm F}(x_1,x_2)=x_1\widehat{x_2}\in\cl_{p,q},
        \quad x_1,x_2\in\cl_{p,q};\nonumber\\
        n=2,&&
        {\rm F}(x_1,x_2)=x_1\widehat{\widetilde{x_2}}\in\cl_{p,q},
        \quad x_1,x_2\in\cl_{p,q};\nonumber\\
        n=3,&&
        {\rm F}(x_1,x_2,x_3,x_4)=x_1\widehat{x_2}\widetilde{x_3}\widehat{\widetilde{x_4}}\in\cl_{p,q},
        \quad x_1,x_2,x_3,x_4\in\cl_{p,q};\nonumber\\
        n=4,&&
        {\rm F}(x_1,x_2,x_3,x_4)=x_1\widehat{\widetilde{x_2}}(\widehat{x_3}\widetilde{x_4})^\bigtriangleup\in\cl_{p,q},
        \quad x_1,x_2,x_3,x_4\in\cl_{p,q};\nonumber\\
        n=5,&&
        {\rm F}(x_1,x_2,\ldots,x_8)=
        x_1\widehat{\widetilde{x_2}}\widehat{x_3}\widetilde{x_4}(\widehat{x_5}\widetilde{x_6}x_7\widehat{\widetilde{x_8}})^\bigtriangleup\in\cl_{p,q},\nonumber\\
        && x_1,x_2,\ldots,x_8\in\cl_{p,q},\nonumber
    \end{eqnarray}
    and $X(k)$ is the set of all possible tuples with $k$ elements $U$ and $N-k$ identity elements $e$:
    \begin{eqnarray}
        X(k)&=&\left\lbrace (x_1,x_2,\ldots,x_N) \, | \,
                    x_i\in\lbrace e,U\rbrace, \quad \sum_{i=1}^N x_i=kU+(N-k)e
                \right\rbrace.\nonumber
    \end{eqnarray}
\end{thm}
\begin{proof}
    The proof is by direct calculation. The presented formulas (\ref{cross<5}) coincide with the formulas from Theorem \ref{theorem2} and Theorem \ref{thm:ck_n=5}.
\end{proof}
Note that the basis-free formulas for the characteristic polynomial coefficients from Theorem~\ref{theorem2} and Theorem~\ref{thm:ck_n=5} look like elementary symmetric polynomials (if we ignore the operation~$\va$) in the variables $U$, $\widehat{U}$, $\widetilde{U}$, etc. The general form of these formulas presented in Theorem~\ref{thm:cross5} illustrates this property.
\begin{rem}\label{review}
    One of the anonymous reviewers noted that all $C_{(k)}$ coefficients for $k < N$ can be computed recursively\footnote{Note that alternatively one could get directly $D_{(k)}(\lambda)=\frac{1}{(N-k)!}\frac{ \partial^{N-k} D_{(N)}(\lambda)}{ \partial\lambda^{N-k}}$.} by differentiating negative determinant $-{\rm Det}(\lambda e - U)$  with respect to $\lambda$ and equating $\lambda$ to zero:
    \begin{eqnarray*}
        &&D_{(N)} (\lambda) := -{\rm Det}(\lambda e -U), \quad U\in\cl_{p,q}, \quad\lambda\in\mathbb{R},\\
        &&D_{(k-1)}(\lambda):=\frac{1}{N-(k-1)}\frac{ \partial D_{(k)}(\lambda)}{ \partial\lambda},
        \quad  k=N,\ldots,1,\\
        &&C_{(k)} = D_{(k)}(0),
    \end{eqnarray*}
    which is a straightforward method to obtain coefficients for any polynomial. Note that this method gives the same result as the method proposed at the beginning of this section.
\end{rem}
\section{The Case $n=6$}
\label{section8}
In this section, we present basis-free formulas for all characteristic polynomial coefficients $C_{(k)}$, $k=1, 2, \ldots, 8$ in the geometric algebras $\cl_{p,q}$, $n=p+q=6$. To obtain the result, we generalize the method from the previous section and apply it to the case $n=6$. We construct the general form of the formulas using the known basis-free formula for the determinant ${\rm Det}$.

The formula
\begin{eqnarray}
&&C_{(8)}
=-{\rm Det}(U)\nonumber\\
&&=
-\frac{1}{3}
U\widetilde{U}\widehat{U}\widehat{\widetilde{U}}(\widehat{U}\widehat{\widetilde{U}}U\widetilde{U})^\bigtriangleup
-\frac{2}{3}
U\widetilde{U}((\widehat{U}\widehat{\widetilde{U}})^\bigtriangleup((\widehat{U}\widehat{\widetilde{U}})^\bigtriangleup(U\widetilde{U})^\bigtriangleup)^\bigtriangleup)^\bigtriangleup,\label{det6}
\end{eqnarray}
is presented in this form in \cite{Shirokov} and in some another form in \cite{Acus}.
The formula 
\begin{eqnarray}
    C_{(1)}={\rm Tr}(U)=U+\widetilde{U}+\widehat{U}+\widehat{\widetilde{U}}+\widehat{U}^\bigtriangleup+\widehat{\widetilde{U}}^\bigtriangleup+U^\bigtriangleup+\widetilde{U}^\bigtriangleup\label{tr6}
\end{eqnarray}
is also presented in \cite{Shirokov}. The formulas for the characteristic polynomial coefficients $C_{(2)}$, $C_{(3)}$, $\ldots$, $C_{(7)}$ are presented for the first time in this paper.

    In the case $n = 6$, the basis-free formulas for the characteristic polynomial coefficients $C_{(k)}\in\cl_{p,q}$, $k = 1,2,\ldots8$ have the  following form
    \begin{eqnarray}
        C_{(k)}=(-1)^{k+1}\sum_{X_j\in X(k)}{\rm F}(X_j),\qquad 
        k = 1,2,\ldots,8, \label{cross6}
    \end{eqnarray}
    where ${\rm F}$ is the function on $8$ variables
    \begin{eqnarray}
        {\rm F}(x_1,x_2,\ldots,x_8)&=&
        \frac{1}{3}x_1\widetilde{x_2}\widehat{x_3}\widehat{\widetilde{x_4}}(\widehat{x_5}\widehat{\widetilde{x_6}}x_7\widetilde{x_8})^\bigtriangleup\nonumber\\
        &&+\frac{2}{3}x_1\widetilde{x_2}((\widehat{x_3}\widehat{\widetilde{x_4}})^\bigtriangleup((\widehat{x_5}\widehat{\widetilde{x_6}})^\bigtriangleup(x_7\widetilde{x_8})^\bigtriangleup)^\bigtriangleup)^\bigtriangleup
        \in\cl_{p,q},\nonumber
    \end{eqnarray}
    \begin{eqnarray}
        x_1,x_2,\ldots,x_8\in\cl_{p,q},\nonumber
    \end{eqnarray}
    $X(k)$ is the set of all possible tuples with $k$ elements $U$ and $8-k$ identity elements $e$:
    \begin{eqnarray}
        X(k)&=&\left\lbrace (x_1,x_2,\ldots,x_8) \, | \,
                    x_i\in\lbrace e,U\rbrace,\quad
                    \sum_{i=1}^8 x_i=kU+(8-k)e
                \right\rbrace.\nonumber
    \end{eqnarray}
 Using numerical Geometric Algebra package for Python \cite{Python}, we checked that the formulas (\ref{cross6}) give valid results for geometric algebra elements with random integer coefficients.
\begin{ex}
    Using (\ref{cross6}), we get
    \begin{eqnarray}
        &&C_{(8)}=
        -\sum_{X_j\in X(8)}{\rm F}(X_j)=-{\rm Det}(U),\nonumber
    \end{eqnarray} 
    which coincides with (\ref{det6}).
    Here $X(8)$ is the set of the only possible tuple with eight elements $U$:
    \begin{eqnarray}
        X(8)&=&\left\lbrace
                \begin{array}{l}
                    (U,U,U,U,U,U,U,U)
                \end{array}
                \right\rbrace.\nonumber
    \end{eqnarray}
\end{ex}

\begin{ex}
    Using (\ref{cross6}), we get
    \begin{eqnarray}
        &&C_{(1)}=
        \sum_{X_j\in X(1)}{\rm F}(X_j)={\rm Tr}(U),\nonumber
    \end{eqnarray} 
    which coincides with (\ref{tr6}).
    Here $X(1)$ is the set of all possible tuples with seven identity elements $e$ and one element $U$:
    \begin{eqnarray}
        X(1)&=&\left\lbrace
                \begin{array}{l}
                    (U,e,e,e,e,e,e,e),
                    (e,U,e,e,e,e,e,e),\\
                    (e,e,U,e,e,e,e,e),
                    (e,e,e,U,e,e,e,e),\\
                    (e,e,e,e,U,e,e,e),
                    (e,e,e,e,e,U,e,e),\\
                    (e,e,e,e,e,e,U,e),
                    (e,e,e,e,e,e,e,U)
                \end{array}
                \right\rbrace.\nonumber
    \end{eqnarray}
\end{ex}

\begin{ex}
    Using (\ref{cross6}), we get
    \begin{eqnarray}
        &&C_{(2)}=
        -\sum_{X_j\in X(2)}{\rm F}(X_j)=
            -(
            (UU^{\bigtriangleup}+
            U\widetilde{U}+
            (\widehat{U}\widehat{\widetilde{U}})^{\bigtriangleup}+
            \widetilde{U}\widehat{\widetilde{U}}+
            \widehat{U}\widehat{\widetilde{U}}
            \nonumber\\&&
            +
            U\widehat{U}^{\bigtriangleup}+\widetilde{U}\widehat{U}^{\bigtriangleup}+
            U\widehat{\widetilde{U}}+
            \widetilde{U}\widetilde{U}^{\bigtriangleup}+
            (U\widetilde{U})^{\bigtriangleup}+
            U\widehat{U}+
            \widetilde{U}\widehat{\widetilde{U}}^{\bigtriangleup}+
            \widetilde{U}\widehat{U}+
            U\widetilde{U}^{\bigtriangleup}
            \nonumber\\&&+
            U\widehat{\widetilde{U}}^{\bigtriangleup}+
            \widetilde{U}U^{\bigtriangleup})+
            \frac{1}{3}(
            (\widehat{\widetilde{U}}\widetilde{U})^{\bigtriangleup}+
            (\widehat{\widetilde{U}}U)^{\bigtriangleup}+
            (\widehat{U}\widetilde{U})^{\bigtriangleup}+
            (\widehat{U}U)^{\bigtriangleup}+
            \widehat{\widetilde{U}}\widetilde{U}^{\bigtriangleup} \nonumber\\&&+
            \widehat{\widetilde{U}}U^{\bigtriangleup}
+ \widehat{\widetilde{U}}\widehat{\widetilde{U}}^{\bigtriangleup}+
            \widehat{\widetilde{U}}\widehat{U}^{\bigtriangleup}+
            \widehat{U}\widetilde{U}^{\bigtriangleup}+
            \widehat{U}U^{\bigtriangleup}+
            \widehat{U}\widehat{\widetilde{U}}^{\bigtriangleup}+
            \widehat{U}\widehat{U}^{\bigtriangleup})
            +\frac{2}{3}( \widehat{\widetilde{U}}^{\bigtriangleup}\widetilde{U}^{\bigtriangleup}\nonumber\\&&+
            \widehat{\widetilde{U}}^{\bigtriangleup}U^{\bigtriangleup}+
            \widehat{U}^{\bigtriangleup}\widetilde{U}^{\bigtriangleup}+
            \widehat{U}^{\bigtriangleup}U^{\bigtriangleup}+
            (\widehat{\widetilde{U}}^{\bigtriangleup}\widetilde{U})^{\bigtriangleup}+
            (\widehat{\widetilde{U}}^{\bigtriangleup}U)^{\bigtriangleup}+
            (\widehat{\widetilde{U}}^{\bigtriangleup}\widehat{\widetilde{U}})^{\bigtriangleup}
            \nonumber\\&&+
            (\widehat{\widetilde{U}}^{\bigtriangleup}\widehat{U})^{\bigtriangleup}+
            (\widehat{U}^{\bigtriangleup}\widetilde{U})^{\bigtriangleup}+
            (\widehat{U}^{\bigtriangleup}U)^{\bigtriangleup}+
            (\widehat{U}^{\bigtriangleup}\widehat{\widetilde{U}})^{\bigtriangleup}+
            (\widehat{U}^{\bigtriangleup}\widehat{U})^{\bigtriangleup})),\nonumber
    \end{eqnarray} 
    where $X(2)$ is the set of all possible tuples with six identity elements $e$ and two elements $U$.
\end{ex}

\begin{ex}
    Using (\ref{cross6}), we get
    \begin{eqnarray*}
        &&C_{(7)}=
        \sum_{X_j\in X(7)}{\rm F}(X_j)=
        \frac{1}{3}(
            \widetilde{U}\widehat{U}\widehat{\widetilde{U}}(\widehat{U}\widehat{\widetilde{U}}U\widetilde{U})^{\bigtriangleup}+
            U\widehat{U}\widehat{\widetilde{U}}(\widehat{U}\widehat{\widetilde{U}}U\widetilde{U})^{\bigtriangleup}\\
            &&+
            U\widetilde{U}\widehat{\widetilde{U}}(\widehat{U}\widehat{\widetilde{U}}U\widetilde{U})^{\bigtriangleup}+
            U\widetilde{U}\widehat{U}(\widehat{U}\widehat{\widetilde{U}}U\widetilde{U})^{\bigtriangleup}+
            U\widetilde{U}\widehat{U}\widehat{\widetilde{U}}(\widehat{\widetilde{U}}U\widetilde{U})^{\bigtriangleup}\\
            &&+
            U\widetilde{U}\widehat{U}\widehat{\widetilde{U}}(\widehat{U}U\widetilde{U})^{\bigtriangleup}+
            U\widetilde{U}\widehat{U}\widehat{\widetilde{U}}(\widehat{U}\widehat{\widetilde{U}}\widetilde{U})^{\bigtriangleup}+
            U\widetilde{U}\widehat{U}\widehat{\widetilde{U}}(\widehat{U}\widehat{\widetilde{U}}U)^{\bigtriangleup})\\
            &&+\frac{2}{3}(
            \widetilde{U}((\widehat{U}\widehat{\widetilde{U}})^{\bigtriangleup}((\widehat{U}\widehat{\widetilde{U}})^{\bigtriangleup}(U\widetilde{U})^{\bigtriangleup})^{\bigtriangleup})^{\bigtriangleup}+
            U((\widehat{U}\widehat{\widetilde{U}})^{\bigtriangleup}((\widehat{U}\widehat{\widetilde{U}})^{\bigtriangleup}(U\widetilde{U})^{\bigtriangleup})^{\bigtriangleup})^{\bigtriangleup}\\
            &&+
            U\widetilde{U}(\widehat{\widetilde{U}}^{\bigtriangleup}((\widehat{U}\widehat{\widetilde{U}})^{\bigtriangleup}(U\widetilde{U})^{\bigtriangleup})^{\bigtriangleup})^{\bigtriangleup}+
            U\widetilde{U}(\widehat{U}^{\bigtriangleup}((\widehat{U}\widehat{\widetilde{U}})^{\bigtriangleup}(U\widetilde{U})^{\bigtriangleup})^{\bigtriangleup})^{\bigtriangleup}\\
            &&+
            U\widetilde{U}((\widehat{U}\widehat{\widetilde{U}})^{\bigtriangleup}(\widehat{\widetilde{U}}^{\bigtriangleup}(U\widetilde{U})^{\bigtriangleup})^{\bigtriangleup})^{\bigtriangleup}+
            U\widetilde{U}((\widehat{U}\widehat{\widetilde{U}})^{\bigtriangleup}(\widehat{U}^{\bigtriangleup}(U\widetilde{U})^{\bigtriangleup})^{\bigtriangleup})^{\bigtriangleup}\\
            &&+
            U\widetilde{U}((\widehat{U}\widehat{\widetilde{U}})^{\bigtriangleup}((\widehat{U}\widehat{\widetilde{U}})^{\bigtriangleup}\widetilde{U}^{\bigtriangleup})^{\bigtriangleup})^{\bigtriangleup}+
            U\widetilde{U}((\widehat{U}\widehat{\widetilde{U}})^{\bigtriangleup}((\widehat{U}\widehat{\widetilde{U}})^{\bigtriangleup}U^{\bigtriangleup})^{\bigtriangleup})^{\bigtriangleup}),\nonumber
    \end{eqnarray*} 
    where $X(7)$ is the set of all possible tuples with one identity element $e$ and seven elements $U$.
\end{ex}
We present basis-free formulas for the remaining characteristic polynomial coefficients $C_{(3)}$, $C_{(4)}$, $C_{(5)}$, $C_{(6)}$ in Appendix~\ref{appendix} because of their cumbersomeness.

\section{Conclusions}

In this paper, for the first time we present the basis-free formulas for all characteristic polynomial coefficients in geometric algebras $\cl_{p,q}$ in the cases $n=p+q=5,6$. These results generalize the results of the paper \cite{Shirokov} for the cases $n\leq 4$. The formulas involve only the operations of geometric product, summation, and operations of conjugation. We actively use the $\bigtriangleup$-conjugation in our considerations. Several new properties of the operation $\va$ and other operations of conjugation in geometric algebra are presented. Using symbolic computation, we verified that the presented basis-free formulas for the characteristic polynomial coefficients in the cases $n\leq 5$ are equivalent to the known recursive formulas (\ref{recur}).  Using numerical Geometric Algebra package for Python \cite{Python}, we checked that the presented basis-free formulas for the characteristic polynomial coefficients in the cases $n=6$ give valid results for geometric algebra elements with random integer coefficients. We present an analytical proof of the basis-free formulas for $C_{(2)}$ and $C_{(3)}$ in the case $n=4$ and the basis-free formula for $C_{(2)}$ in the case $n=5$. The proof of the equivalence of the proposed basis-free formulas to the recursive ones turned out to be rather nontrivial.

We provide important special cases of the basis-free formulas. For elements of grade 1 (vectors) and the basis elements, we found the basis-free formulas for all characteristic polynomial coefficients in the case of arbitrary $n$. For elements of group ${\rm Spin}(p,q)$, $p+q=n\leq 5$, we present the basis-free formulas in Theorem \ref{thm:spin}. Using specific examples, we show the significance of the theorem conditions and the difference between Theorem~\ref{lemma:c_2s-1=0} and Theorem~\ref{thm:like_basis_el}.

We introduce the method of obtaining the general form of the presented basis-free formulas for the characteristic polynomial coefficients using the basis-free formula for ${\rm Det}$ (determinant) in the cases $n\leq 6$. Using this method, we solve the dimensions $n=5$ and $n=6$. The applicability of the method to the higher dimensions is a subject for further research.

Different (recursive and explicit) formulas can be used for different purposes. The recursive formulas are interesting from theoretical and computational points of view. We actively use the recursive formulas in our analytical proofs. The explicit basis-free formulas for the characteristic polynomial coefficients in the cases $n\leq 6$ look like elementary symmetric polynomials (if we ignore the operation $\va$) in the variables $U$, $\widehat{U}$, $\widetilde{U}$, etc. (see the details in Section~\ref{section7} and Section~\ref{section8}). This observation is interesting from a theoretical point of view. The presented explicit formulas for characteristic polynomial coefficients allow us to obtain the simplified explicit formulas for some special cases, in particular, for elements of spin groups (see Section~\ref{section6}).  
 
The geometric algebras of vector spaces of dimensions $n=4,5,$ and $6$ are important for different applications in physics (the space-time algebra $\cl_{1,3}$ \cite{Doran,Hestenes,Lasenby}, the conformal space-time algebras $\cl_{4,2}$ and $\cl_{2,4}$ \cite{Dirac,Doran}), in computer science and engineering (the conformal geometric algebra $\cl_{4,1}$ \cite{Bayro,Breuils,Dorst,Hildenbrand,Li}), in computer vision and computer graphics (the geometric algebra $\cl_{3,3}$ of projective geometry \cite{Dorst2,Klawitter}). In particular, the characteristic polynomial coefficients are used to solve the Sylvester and Lyapunov equations in geometric algebra \cite{CGI20,CGI20_extend}. The presented basis-free formulas for characteristic polynomial coefficients can also be used in symbolic computation using different software \cite{AcusPack,Python,HitzerMatlab,Python2}.

\section*{Acknowledgment}

The main results of this paper were reported at the International Conference “Computer Graphics
International 2021 (CGI2021)” (online, Geneva, Switzerland, September 2021). The authors are grateful to the organizers and the participants of this conference for fruitful discussions.

The authors are grateful to the anonymous Reviewers for their careful reading of the paper and helpful comments on how to improve the presentation.

The publication was prepared within the framework of the Academic Fund Program at the HSE University in 2022 (grant 22-00-001).

\medskip

\noindent{\bf Data availability} Data sharing not applicable to this article as no datasets were generated or analyzed during the current study.

\noindent\textbf{Declarations}

\noindent\textbf{Conflict of Interest} The authors declare that they have no conflict of interest.

\appendix

\section{Basis-Free Formulas in the Case $n=6$}
\label{appendix}
In the case $n=6$, we have the following basis-free formulas for the characteristic polynomial coefficients $C_{(3)}, C_{(4)}, C_{(5)}, C_{(6)}\in\cl_{p,q}$. The remaining characteristic polynomial coefficients $C_{(1)}, C_{(2)}, C_{(7)}, C_{(8)}\in\cl_{p,q}$ are presented in Section~\ref{section8}.
    {\footnotesize
    \\\\
    $C_{(3)}=
            (\widetilde{U}(\widehat{U}\widehat{\widetilde{U}})^{\bigtriangleup}+
            U\widetilde{U}\widehat{\widetilde{U}}^{\bigtriangleup}\!\!\!\!+
            U\widetilde{U}\widehat{\widetilde{U}}+
            U(\widehat{U}\widehat{\widetilde{U}})^{\bigtriangleup}
            +
            U\widetilde{U}\widehat{U}^{\bigtriangleup}+
            U\widehat{U}\widehat{\widetilde{U}}+
            U\widetilde{U}\widetilde{U}^{\bigtriangleup}
            +
            U\widetilde{U}U^{\bigtriangleup}
            +
            U\widetilde{U}\widehat{U}+
            \widetilde{U}(U\widetilde{U})^{\bigtriangleup}+
            U(U\widetilde{U})^{\bigtriangleup}+
            \widetilde{U}\widehat{U}\widehat{\widetilde{U}})
            +\frac{1}{3}(
            (\widehat{\widetilde{U}}U\widetilde{U})^{\bigtriangleup}+
            (\widehat{U}U\widetilde{U})^{\bigtriangleup}
            +
            (\widehat{U}\widehat{\widetilde{U}}\widetilde{U})^{\bigtriangleup}+
            (\widehat{U}\widehat{\widetilde{U}}U)^{\bigtriangleup}
            +
            \widehat{\widetilde{U}}(U\widetilde{U})^{\bigtriangleup}+
            \widehat{\widetilde{U}}(\widehat{\widetilde{U}}\widetilde{U})^{\bigtriangleup}+
            \widehat{\widetilde{U}}(\widehat{\widetilde{U}}U)^{\bigtriangleup}+
            \widehat{\widetilde{U}}(\widehat{U}\widetilde{U})^{\bigtriangleup}
            +
            \widehat{\widetilde{U}}(\widehat{U}U)^{\bigtriangleup}
            +
            \widehat{\widetilde{U}}(\widehat{U}\widehat{\widetilde{U}})^{\bigtriangleup}+
            \widehat{U}(U\widetilde{U})^{\bigtriangleup}+
            \widehat{U}(\widehat{\widetilde{U}}\widetilde{U})^{\bigtriangleup}
            +
            \widehat{U}(\widehat{\widetilde{U}}U)^{\bigtriangleup}+
            \widehat{U}(\widehat{U}\widetilde{U})^{\bigtriangleup}
            +
            \widehat{U}(\widehat{U}U)^{\bigtriangleup}
            +
            \widehat{U}(\widehat{U}\widehat{\widetilde{U}})^{\bigtriangleup}
            +
            \widehat{U}\widehat{\widetilde{U}}\widetilde{U}^{\bigtriangleup}+
            \widehat{U}\widehat{\widetilde{U}}U^{\bigtriangleup}+
            \widehat{U}\widehat{\widetilde{U}}\widehat{\widetilde{U}}^{\bigtriangleup}+
            \widehat{U}\widehat{\widetilde{U}}\widehat{U}^{\bigtriangleup}
            +
            \widetilde{U}(\widehat{\widetilde{U}}\widetilde{U})^{\bigtriangleup}+
            \widetilde{U}(\widehat{\widetilde{U}}U)^{\bigtriangleup}+
            \widetilde{U}(\widehat{U}\widetilde{U})^{\bigtriangleup}
            +
            \widetilde{U}(\widehat{U}U)^{\bigtriangleup}
            +
            \widetilde{U}\widehat{\widetilde{U}}\widetilde{U}^{\bigtriangleup}
            +
            \widetilde{U}\widehat{\widetilde{U}}U^{\bigtriangleup}+
            \widetilde{U}\widehat{\widetilde{U}}\widehat{\widetilde{U}}^{\bigtriangleup}+
            \widetilde{U}\widehat{\widetilde{U}}\widehat{U}^{\bigtriangleup}
            +
            \widetilde{U}\widehat{U}\widetilde{U}^{\bigtriangleup}+
            \widetilde{U}\widehat{U}U^{\bigtriangleup}
            +
            \widetilde{U}\widehat{U}\widehat{\widetilde{U}}^{\bigtriangleup}
            +
            \widetilde{U}\widehat{U}\widehat{U}^{\bigtriangleup}
            +
            U(\widehat{\widetilde{U}}\widetilde{U})^{\bigtriangleup}+
            U(\widehat{\widetilde{U}}U)^{\bigtriangleup}+
            U(\widehat{U}\widetilde{U})^{\bigtriangleup}+
            U(\widehat{U}U)^{\bigtriangleup}
            +
            U\widehat{\widetilde{U}}\widetilde{U}^{\bigtriangleup}
            +
            U\widehat{\widetilde{U}}U^{\bigtriangleup}+
            U\widehat{\widetilde{U}}\widehat{\widetilde{U}}^{\bigtriangleup}+
            U\widehat{\widetilde{U}}\widehat{U}^{\bigtriangleup}
            +
            U\widehat{U}\widetilde{U}^{\bigtriangleup}+
            U\widehat{U}U^{\bigtriangleup}+
            U\widehat{U}\widehat{\widetilde{U}}^{\bigtriangleup}+
            U\widehat{U}\widehat{U}^{\bigtriangleup})
            +\frac{2}{3}(
            \widehat{\widetilde{U}}^{\bigtriangleup}(U\widetilde{U})^{\bigtriangleup}+
            \widehat{U}^{\bigtriangleup}(U\widetilde{U})^{\bigtriangleup}+
            (\widehat{U}\widehat{\widetilde{U}})^{\bigtriangleup}\widetilde{U}^{\bigtriangleup}
            +
            (\widehat{U}\widehat{\widetilde{U}})^{\bigtriangleup}U^{\bigtriangleup}
            +
            (\widehat{\widetilde{U}}^{\bigtriangleup}U\widetilde{U})^{\bigtriangleup}+
            (\widehat{\widetilde{U}}^{\bigtriangleup}(\widehat{\widetilde{U}}^{\bigtriangleup}\widetilde{U}^{\bigtriangleup})^{\bigtriangleup})^{\bigtriangleup}+
            (\widehat{\widetilde{U}}^{\bigtriangleup}(\widehat{\widetilde{U}}^{\bigtriangleup}U^{\bigtriangleup})^{\bigtriangleup})^{\bigtriangleup}+
            (\widehat{\widetilde{U}}^{\bigtriangleup}(\widehat{U}^{\bigtriangleup}\widetilde{U}^{\bigtriangleup})^{\bigtriangleup})^{\bigtriangleup}
            +
            (\widehat{\widetilde{U}}^{\bigtriangleup}(\widehat{U}^{\bigtriangleup}U^{\bigtriangleup})^{\bigtriangleup})^{\bigtriangleup}+
            (\widehat{\widetilde{U}}^{\bigtriangleup}\widehat{U}\widehat{\widetilde{U}})^{\bigtriangleup}+
            (\widehat{U}^{\bigtriangleup}U\widetilde{U})^{\bigtriangleup}+
            (\widehat{U}^{\bigtriangleup}(\widehat{\widetilde{U}}^{\bigtriangleup}\widetilde{U}^{\bigtriangleup})^{\bigtriangleup})^{\bigtriangleup}
            +
            (\widehat{U}^{\bigtriangleup}(\widehat{\widetilde{U}}^{\bigtriangleup}U^{\bigtriangleup})^{\bigtriangleup})^{\bigtriangleup}
            +
            (\widehat{U}^{\bigtriangleup}(\widehat{U}^{\bigtriangleup}\widetilde{U}^{\bigtriangleup})^{\bigtriangleup})^{\bigtriangleup}+
            (\widehat{U}^{\bigtriangleup}(\widehat{U}^{\bigtriangleup}U^{\bigtriangleup})^{\bigtriangleup})^{\bigtriangleup}+
            (\widehat{U}^{\bigtriangleup}\widehat{U}\widehat{\widetilde{U}})^{\bigtriangleup}
            +
            ((\widehat{U}\widehat{\widetilde{U}})^{\bigtriangleup}\widetilde{U})^{\bigtriangleup}+
            ((\widehat{U}\widehat{\widetilde{U}})^{\bigtriangleup}U)^{\bigtriangleup}
            +
            ((\widehat{U}\widehat{\widetilde{U}})^{\bigtriangleup}\widehat{\widetilde{U}})^{\bigtriangleup}
            +
            ((\widehat{U}\widehat{\widetilde{U}})^{\bigtriangleup}\widehat{U})^{\bigtriangleup}
            +
            \widetilde{U}\widehat{\widetilde{U}}^{\bigtriangleup}\widetilde{U}^{\bigtriangleup}+
            \widetilde{U}\widehat{\widetilde{U}}^{\bigtriangleup}U^{\bigtriangleup}+
            \widetilde{U}\widehat{U}^{\bigtriangleup}\widetilde{U}^{\bigtriangleup}+
            \widetilde{U}\widehat{U}^{\bigtriangleup}U^{\bigtriangleup}
            +
            \widetilde{U}(\widehat{\widetilde{U}}^{\bigtriangleup}\widetilde{U})^{\bigtriangleup}
            +
            \widetilde{U}(\widehat{\widetilde{U}}^{\bigtriangleup}U)^{\bigtriangleup}+
            \widetilde{U}(\widehat{\widetilde{U}}^{\bigtriangleup}\widehat{\widetilde{U}})^{\bigtriangleup}+
            \widetilde{U}(\widehat{\widetilde{U}}^{\bigtriangleup}\widehat{U})^{\bigtriangleup}
            +
            \widetilde{U}(\widehat{U}^{\bigtriangleup}\widetilde{U})^{\bigtriangleup}+
            \widetilde{U}(\widehat{U}^{\bigtriangleup}U)^{\bigtriangleup}+
            \widetilde{U}(\widehat{U}^{\bigtriangleup}\widehat{\widetilde{U}})^{\bigtriangleup}
            +
            \widetilde{U}(\widehat{U}^{\bigtriangleup}\widehat{U})^{\bigtriangleup}
            +
            U\widehat{\widetilde{U}}^{\bigtriangleup}\widetilde{U}^{\bigtriangleup}+
            U\widehat{\widetilde{U}}^{\bigtriangleup}U^{\bigtriangleup}+
            U\widehat{U}^{\bigtriangleup}\widetilde{U}^{\bigtriangleup}+
            U\widehat{U}^{\bigtriangleup}U^{\bigtriangleup}
            +
            U(\widehat{\widetilde{U}}^{\bigtriangleup}\widetilde{U})^{\bigtriangleup}+
            U(\widehat{\widetilde{U}}^{\bigtriangleup}U)^{\bigtriangleup}
            +
            U(\widehat{\widetilde{U}}^{\bigtriangleup}\widehat{\widetilde{U}})^{\bigtriangleup}
            +
            U(\widehat{\widetilde{U}}^{\bigtriangleup}\widehat{U})^{\bigtriangleup}
            +
            U(\widehat{U}^{\bigtriangleup}\widetilde{U})^{\bigtriangleup}+
            U(\widehat{U}^{\bigtriangleup}U)^{\bigtriangleup}+
            U(\widehat{U}^{\bigtriangleup}\widehat{\widetilde{U}})^{\bigtriangleup}+
            U(\widehat{U}^{\bigtriangleup}\widehat{U})^{\bigtriangleup}),$
            \\\\
    $C_{(4)}=-(
            (U\widetilde{U}\widehat{U}\widehat{\widetilde{U}}+
            U\widetilde{U}(\widehat{U}\widehat{\widetilde{U}})^{\bigtriangleup}+
            U\widetilde{U}(U\widetilde{U})^{\bigtriangleup})
            +\frac{1}{3}(
            (\widehat{U}\widehat{\widetilde{U}}U\widetilde{U})^{\bigtriangleup}+
            \widehat{\widetilde{U}}(\widehat{\widetilde{U}}U\widetilde{U})^{\bigtriangleup}+
            \widehat{\widetilde{U}}(\widehat{U}U\widetilde{U})^{\bigtriangleup}+
            \widehat{\widetilde{U}}(\widehat{U}\widehat{\widetilde{U}}\widetilde{U})^{\bigtriangleup}
            +
            \widehat{\widetilde{U}}(\widehat{U}\widehat{\widetilde{U}}U)^{\bigtriangleup}+
            \widehat{U}(\widehat{\widetilde{U}}U\widetilde{U})^{\bigtriangleup}+
            \widehat{U}(\widehat{U}U\widetilde{U})^{\bigtriangleup}+
            \widehat{U}(\widehat{U}\widehat{\widetilde{U}}\widetilde{U})^{\bigtriangleup}
            +
            \widehat{U}(\widehat{U}\widehat{\widetilde{U}}U)^{\bigtriangleup}+
            \widehat{U}\widehat{\widetilde{U}}(U\widetilde{U})^{\bigtriangleup}+
            \widehat{U}\widehat{\widetilde{U}}(\widehat{\widetilde{U}}\widetilde{U})^{\bigtriangleup}+
            \widehat{U}\widehat{\widetilde{U}}(\widehat{\widetilde{U}}U)^{\bigtriangleup}
            +
            \widehat{U}\widehat{\widetilde{U}}(\widehat{U}\widetilde{U})^{\bigtriangleup}+
            \widehat{U}\widehat{\widetilde{U}}(\widehat{U}U)^{\bigtriangleup}+
            \widehat{U}\widehat{\widetilde{U}}(\widehat{U}\widehat{\widetilde{U}})^{\bigtriangleup}+
            \widetilde{U}(\widehat{\widetilde{U}}U\widetilde{U})^{\bigtriangleup}
            +
            \widetilde{U}(\widehat{U}U\widetilde{U})^{\bigtriangleup}+
            \widetilde{U}(\widehat{U}\widehat{\widetilde{U}}\widetilde{U})^{\bigtriangleup}+
            \widetilde{U}(\widehat{U}\widehat{\widetilde{U}}U)^{\bigtriangleup}+
            \widetilde{U}\widehat{\widetilde{U}}(U\widetilde{U})^{\bigtriangleup}
            +
            \widetilde{U}\widehat{\widetilde{U}}(\widehat{\widetilde{U}}\widetilde{U})^{\bigtriangleup}+
            \widetilde{U}\widehat{\widetilde{U}}(\widehat{\widetilde{U}}U)^{\bigtriangleup}+
            \widetilde{U}\widehat{\widetilde{U}}(\widehat{U}\widetilde{U})^{\bigtriangleup}+
            \widetilde{U}\widehat{\widetilde{U}}(\widehat{U}U)^{\bigtriangleup}
            +
            \widetilde{U}\widehat{\widetilde{U}}(\widehat{U}\widehat{\widetilde{U}})^{\bigtriangleup}+
            \widetilde{U}\widehat{U}(U\widetilde{U})^{\bigtriangleup}+
            \widetilde{U}\widehat{U}(\widehat{\widetilde{U}}\widetilde{U})^{\bigtriangleup}+
            \widetilde{U}\widehat{U}(\widehat{\widetilde{U}}U)^{\bigtriangleup}
            +
            \widetilde{U}\widehat{U}(\widehat{U}\widetilde{U})^{\bigtriangleup}+
            \widetilde{U}\widehat{U}(\widehat{U}U)^{\bigtriangleup}+
            \widetilde{U}\widehat{U}(\widehat{U}\widehat{\widetilde{U}})^{\bigtriangleup}+
            \widetilde{U}\widehat{U}\widehat{\widetilde{U}}\widetilde{U}^{\bigtriangleup}
            +
            \widetilde{U}\widehat{U}\widehat{\widetilde{U}}U^{\bigtriangleup}+
            \widetilde{U}\widehat{U}\widehat{\widetilde{U}}\widehat{\widetilde{U}}^{\bigtriangleup}+
            \widetilde{U}\widehat{U}\widehat{\widetilde{U}}\widehat{U}^{\bigtriangleup}+
            U(\widehat{\widetilde{U}}U\widetilde{U})^{\bigtriangleup}
            +
            U(\widehat{U}U\widetilde{U})^{\bigtriangleup}+
            U(\widehat{U}\widehat{\widetilde{U}}\widetilde{U})^{\bigtriangleup}+
            U(\widehat{U}\widehat{\widetilde{U}}U)^{\bigtriangleup}+
            U\widehat{\widetilde{U}}(U\widetilde{U})^{\bigtriangleup}
            +
            U\widehat{\widetilde{U}}(\widehat{\widetilde{U}}\widetilde{U})^{\bigtriangleup}+
            U\widehat{\widetilde{U}}(\widehat{\widetilde{U}}U)^{\bigtriangleup}+
            U\widehat{\widetilde{U}}(\widehat{U}\widetilde{U})^{\bigtriangleup}+
            U\widehat{\widetilde{U}}(\widehat{U}U)^{\bigtriangleup}
            +
            U\widehat{\widetilde{U}}(\widehat{U}\widehat{\widetilde{U}})^{\bigtriangleup}+
            U\widehat{U}(U\widetilde{U})^{\bigtriangleup}+
            U\widehat{U}(\widehat{\widetilde{U}}\widetilde{U})^{\bigtriangleup}+
            U\widehat{U}(\widehat{\widetilde{U}}U)^{\bigtriangleup}
            +
            U\widehat{U}(\widehat{U}\widetilde{U})^{\bigtriangleup}+
            U\widehat{U}(\widehat{U}U)^{\bigtriangleup}+
            U\widehat{U}(\widehat{U}\widehat{\widetilde{U}})^{\bigtriangleup}+
            U\widehat{U}\widehat{\widetilde{U}}\widetilde{U}^{\bigtriangleup}
            +
            U\widehat{U}\widehat{\widetilde{U}}U^{\bigtriangleup}+
            U\widehat{U}\widehat{\widetilde{U}}\widehat{\widetilde{U}}^{\bigtriangleup}+
            U\widehat{U}\widehat{\widetilde{U}}\widehat{U}^{\bigtriangleup}+
            U\widetilde{U}(\widehat{\widetilde{U}}\widetilde{U})^{\bigtriangleup}
            +
            U\widetilde{U}(\widehat{\widetilde{U}}U)^{\bigtriangleup}+
            U\widetilde{U}(\widehat{U}\widetilde{U})^{\bigtriangleup}+
            U\widetilde{U}(\widehat{U}U)^{\bigtriangleup}+
            U\widetilde{U}\widehat{\widetilde{U}}\widetilde{U}^{\bigtriangleup}
            +
            U\widetilde{U}\widehat{\widetilde{U}}U^{\bigtriangleup}+
            U\widetilde{U}\widehat{\widetilde{U}}\widehat{\widetilde{U}}^{\bigtriangleup}+
            U\widetilde{U}\widehat{\widetilde{U}}\widehat{U}^{\bigtriangleup}+
            U\widetilde{U}\widehat{U}\widetilde{U}^{\bigtriangleup}
            +
            U\widetilde{U}\widehat{U}U^{\bigtriangleup}+
            U\widetilde{U}\widehat{U}\widehat{\widetilde{U}}^{\bigtriangleup}+
            U\widetilde{U}\widehat{U}\widehat{U}^{\bigtriangleup})
            +\frac{2}{3}(
            (\widehat{U}\widehat{\widetilde{U}})^{\bigtriangleup}(U\widetilde{U})^{\bigtriangleup}+
            (\widehat{\widetilde{U}}^{\bigtriangleup}\!\!(\widehat{\widetilde{U}}^{\bigtriangleup}\!\!(U\widetilde{U})^{\bigtriangleup})^{\bigtriangleup})^{\bigtriangleup}+
            (\widehat{\widetilde{U}}^{\bigtriangleup}\!\!(\widehat{U}^{\bigtriangleup}(U\widetilde{U})^{\bigtriangleup})^{\bigtriangleup})^{\bigtriangleup}
            +
            (\widehat{\widetilde{U}}^{\bigtriangleup}((\widehat{U}\widehat{\widetilde{U}})^{\bigtriangleup}\widetilde{U}^{\bigtriangleup})^{\bigtriangleup})^{\bigtriangleup}+
            (\widehat{\widetilde{U}}^{\bigtriangleup}((\widehat{U}\widehat{\widetilde{U}})^{\bigtriangleup}U^{\bigtriangleup})^{\bigtriangleup})^{\bigtriangleup}+
            (\widehat{U}^{\bigtriangleup}(\widehat{\widetilde{U}}^{\bigtriangleup}(U\widetilde{U})^{\bigtriangleup})^{\bigtriangleup})^{\bigtriangleup}
            +
            (\widehat{U}^{\bigtriangleup}(\widehat{U}^{\bigtriangleup}(U\widetilde{U})^{\bigtriangleup})^{\bigtriangleup})^{\bigtriangleup}+\\
            (\widehat{U}^{\bigtriangleup}((\widehat{U}\widehat{\widetilde{U}})^{\bigtriangleup}\widetilde{U}^{\bigtriangleup})^{\bigtriangleup})^{\bigtriangleup}
            +
            (\widehat{U}^{\bigtriangleup}((\widehat{U}\widehat{\widetilde{U}})^{\bigtriangleup}U^{\bigtriangleup})^{\bigtriangleup})^{\bigtriangleup}
            +
            ((\widehat{U}\widehat{\widetilde{U}})^{\bigtriangleup}U\widetilde{U})^{\bigtriangleup}
            +
            ((\widehat{U}\widehat{\widetilde{U}})^{\bigtriangleup}(\widehat{\widetilde{U}}^{\bigtriangleup}\widetilde{U}^{\bigtriangleup})^{\bigtriangleup})^{\bigtriangleup}
            +
            ((\widehat{U}\widehat{\widetilde{U}})^{\bigtriangleup}(\widehat{\widetilde{U}}^{\bigtriangleup}U^{\bigtriangleup})^{\bigtriangleup})^{\bigtriangleup}
            +
            ((\widehat{U}\widehat{\widetilde{U}})^{\bigtriangleup}(\widehat{U}^{\bigtriangleup}\widetilde{U}^{\bigtriangleup})^{\bigtriangleup})^{\bigtriangleup}
            +
            ((\widehat{U}\widehat{\widetilde{U}})^{\bigtriangleup}(\widehat{U}^{\bigtriangleup}U^{\bigtriangleup})^{\bigtriangleup})^{\bigtriangleup}
            +
            ((\widehat{U}\widehat{\widetilde{U}})^{\bigtriangleup}\widehat{U}\widehat{\widetilde{U}})^{\bigtriangleup}
            +
            \widetilde{U}\widehat{\widetilde{U}}^{\bigtriangleup}(U\widetilde{U})^{\bigtriangleup}
            +
            \widetilde{U}\widehat{U}^{\bigtriangleup}(U\widetilde{U})^{\bigtriangleup}
            +
            \widetilde{U}(\widehat{U}\widehat{\widetilde{U}})^{\bigtriangleup}\widetilde{U}^{\bigtriangleup}
            +
            \widetilde{U}(\widehat{U}\widehat{\widetilde{U}})^{\bigtriangleup}U^{\bigtriangleup}
            +
            \widetilde{U}(\widehat{\widetilde{U}}^{\bigtriangleup}U\widetilde{U})^{\bigtriangleup}
            +\\
            \widetilde{U}(\widehat{\widetilde{U}}^{\bigtriangleup}(\widehat{\widetilde{U}}^{\bigtriangleup}\widetilde{U}^{\bigtriangleup})^{\bigtriangleup})^{\bigtriangleup}
            +
            \widetilde{U}(\widehat{\widetilde{U}}^{\bigtriangleup}(\widehat{\widetilde{U}}^{\bigtriangleup}U^{\bigtriangleup})^{\bigtriangleup})^{\bigtriangleup}
            +
            \widetilde{U}(\widehat{\widetilde{U}}^{\bigtriangleup}(\widehat{U}^{\bigtriangleup}\widetilde{U}^{\bigtriangleup})^{\bigtriangleup})^{\bigtriangleup}
            +
            \widetilde{U}(\widehat{\widetilde{U}}^{\bigtriangleup}(\widehat{U}^{\bigtriangleup}U^{\bigtriangleup})^{\bigtriangleup})^{\bigtriangleup}
            +
            \widetilde{U}(\widehat{\widetilde{U}}^{\bigtriangleup}\widehat{U}\widehat{\widetilde{U}})^{\bigtriangleup}
            +
            \widetilde{U}(\widehat{U}^{\bigtriangleup}U\widetilde{U})^{\bigtriangleup}
            +
            \widetilde{U}(\widehat{U}^{\bigtriangleup}(\widehat{\widetilde{U}}^{\bigtriangleup}\widetilde{U}^{\bigtriangleup})^{\bigtriangleup})^{\bigtriangleup}
            +
            \widetilde{U}(\widehat{U}^{\bigtriangleup}(\widehat{\widetilde{U}}^{\bigtriangleup}U^{\bigtriangleup})^{\bigtriangleup})^{\bigtriangleup}
            +\\
            \widetilde{U}(\widehat{U}^{\bigtriangleup}(\widehat{U}^{\bigtriangleup}\widetilde{U}^{\bigtriangleup})^{\bigtriangleup})^{\bigtriangleup}
            +
            \widetilde{U}(\widehat{U}^{\bigtriangleup}(\widehat{U}^{\bigtriangleup}U^{\bigtriangleup})^{\bigtriangleup})^{\bigtriangleup}
            +
            \widetilde{U}(\widehat{U}^{\bigtriangleup}\widehat{U}\widehat{\widetilde{U}})^{\bigtriangleup}
            +
            \widetilde{U}((\widehat{U}\widehat{\widetilde{U}})^{\bigtriangleup}\widetilde{U})^{\bigtriangleup}
            +\\
            \widetilde{U}((\widehat{U}\widehat{\widetilde{U}})^{\bigtriangleup}U)^{\bigtriangleup}
            +
            \widetilde{U}((\widehat{U}\widehat{\widetilde{U}})^{\bigtriangleup}\widehat{\widetilde{U}})^{\bigtriangleup}
            +
            \widetilde{U}((\widehat{U}\widehat{\widetilde{U}})^{\bigtriangleup}\widehat{U})^{\bigtriangleup}
            +
            U\widehat{\widetilde{U}}^{\bigtriangleup}(U\widetilde{U})^{\bigtriangleup}
            +
            U\widehat{U}^{\bigtriangleup}(U\widetilde{U})^{\bigtriangleup}
            +\\
            U(\widehat{U}\widehat{\widetilde{U}})^{\bigtriangleup}\widetilde{U}^{\bigtriangleup}
            +
            U(\widehat{U}\widehat{\widetilde{U}})^{\bigtriangleup}U^{\bigtriangleup}
            +
            U(\widehat{\widetilde{U}}^{\bigtriangleup}U\widetilde{U})^{\bigtriangleup}
            +
            U(\widehat{\widetilde{U}}^{\bigtriangleup}(\widehat{\widetilde{U}}^{\bigtriangleup}\widetilde{U}^{\bigtriangleup})^{\bigtriangleup})^{\bigtriangleup}
            +
            U(\widehat{\widetilde{U}}^{\bigtriangleup}(\widehat{\widetilde{U}}^{\bigtriangleup}U^{\bigtriangleup})^{\bigtriangleup})^{\bigtriangleup}
            +
            U(\widehat{\widetilde{U}}^{\bigtriangleup}(\widehat{U}^{\bigtriangleup}\widetilde{U}^{\bigtriangleup})^{\bigtriangleup})^{\bigtriangleup}
            +
            U(\widehat{\widetilde{U}}^{\bigtriangleup}(\widehat{U}^{\bigtriangleup}U^{\bigtriangleup})^{\bigtriangleup})^{\bigtriangleup}
            +
            U(\widehat{\widetilde{U}}^{\bigtriangleup}\widehat{U}\widehat{\widetilde{U}})^{\bigtriangleup}
            +
            U(\widehat{U}^{\bigtriangleup}U\widetilde{U})^{\bigtriangleup}
            +\\
            U(\widehat{U}^{\bigtriangleup}(\widehat{\widetilde{U}}^{\bigtriangleup}\widetilde{U}^{\bigtriangleup})^{\bigtriangleup})^{\bigtriangleup}
            +
            U(\widehat{U}^{\bigtriangleup}(\widehat{\widetilde{U}}^{\bigtriangleup}U^{\bigtriangleup})^{\bigtriangleup})^{\bigtriangleup}
            +
            U(\widehat{U}^{\bigtriangleup}(\widehat{U}^{\bigtriangleup}\widetilde{U}^{\bigtriangleup})^{\bigtriangleup})^{\bigtriangleup}
            +
            U(\widehat{U}^{\bigtriangleup}(\widehat{U}^{\bigtriangleup}U^{\bigtriangleup})^{\bigtriangleup})^{\bigtriangleup}
            +
            U(\widehat{U}^{\bigtriangleup}\widehat{U}\widehat{\widetilde{U}})^{\bigtriangleup}
            +
            U((\widehat{U}\widehat{\widetilde{U}})^{\bigtriangleup}\widetilde{U})^{\bigtriangleup}
            +
            U((\widehat{U}\widehat{\widetilde{U}})^{\bigtriangleup}U)^{\bigtriangleup}
            +
            U((\widehat{U}\widehat{\widetilde{U}})^{\bigtriangleup}\widehat{\widetilde{U}})^{\bigtriangleup}
            +
            U((\widehat{U}\widehat{\widetilde{U}})^{\bigtriangleup}\widehat{U})^{\bigtriangleup}
            +\\
            U\widetilde{U}\widehat{\widetilde{U}}^{\bigtriangleup}\widetilde{U}^{\bigtriangleup}
            +
            U\widetilde{U}\widehat{\widetilde{U}}^{\bigtriangleup}U^{\bigtriangleup}
            +
            U\widetilde{U}\widehat{U}^{\bigtriangleup}\widetilde{U}^{\bigtriangleup}+
            U\widetilde{U}\widehat{U}^{\bigtriangleup}U^{\bigtriangleup}+
            U\widetilde{U}(\widehat{\widetilde{U}}^{\bigtriangleup}\widetilde{U})^{\bigtriangleup}
            +
            U\widetilde{U}(\widehat{\widetilde{U}}^{\bigtriangleup}U)^{\bigtriangleup}+
            U\widetilde{U}(\widehat{\widetilde{U}}^{\bigtriangleup}\widehat{\widetilde{U}})^{\bigtriangleup}+
            U\widetilde{U}(\widehat{\widetilde{U}}^{\bigtriangleup}\widehat{U})^{\bigtriangleup}
            +
            U\widetilde{U}(\widehat{U}^{\bigtriangleup}\widetilde{U})^{\bigtriangleup}+
            U\widetilde{U}(\widehat{U}^{\bigtriangleup}U)^{\bigtriangleup}+
            U\widetilde{U}(\widehat{U}^{\bigtriangleup}\widehat{\widetilde{U}})^{\bigtriangleup}
            +
            U\widetilde{U}(\widehat{U}^{\bigtriangleup}\widehat{U})^{\bigtriangleup})),
            $\\\\
        $C_{(5)}=
            \frac{1}{3}(\widehat{\widetilde{U}}(\widehat{U}\widehat{\widetilde{U}}U\widetilde{U})^{\bigtriangleup}+
            \widehat{U}(\widehat{U}\widehat{\widetilde{U}}U\widetilde{U})^{\bigtriangleup}+
            \widehat{U}\widehat{\widetilde{U}}(\widehat{\widetilde{U}}U\widetilde{U})^{\bigtriangleup}
            +
            \widehat{U}\widehat{\widetilde{U}}(\widehat{U}U\widetilde{U})^{\bigtriangleup}+
            \widehat{U}\widehat{\widetilde{U}}(\widehat{U}\widehat{\widetilde{U}}\widetilde{U})^{\bigtriangleup}+
            \widehat{U}\widehat{\widetilde{U}}(\widehat{U}\widehat{\widetilde{U}}U)^{\bigtriangleup}
            +
            \widetilde{U}(\widehat{U}\widehat{\widetilde{U}}U\widetilde{U})^{\bigtriangleup}+
            \widetilde{U}\widehat{\widetilde{U}}(\widehat{\widetilde{U}}U\widetilde{U})^{\bigtriangleup}+
            \widetilde{U}\widehat{\widetilde{U}}(\widehat{U}U\widetilde{U})^{\bigtriangleup}
            +
            \widetilde{U}\widehat{\widetilde{U}}(\widehat{U}\widehat{\widetilde{U}}\widetilde{U})^{\bigtriangleup}+
            \widetilde{U}\widehat{\widetilde{U}}(\widehat{U}\widehat{\widetilde{U}}U)^{\bigtriangleup}+
            \widetilde{U}\widehat{U}(\widehat{\widetilde{U}}U\widetilde{U})^{\bigtriangleup}
            +
            \widetilde{U}\widehat{U}(\widehat{U}U\widetilde{U})^{\bigtriangleup}+
            \widetilde{U}\widehat{U}(\widehat{U}\widehat{\widetilde{U}}\widetilde{U})^{\bigtriangleup}+
            \widetilde{U}\widehat{U}(\widehat{U}\widehat{\widetilde{U}}U)^{\bigtriangleup}
            +
            \widetilde{U}\widehat{U}\widehat{\widetilde{U}}(U\widetilde{U})^{\bigtriangleup}+
            \widetilde{U}\widehat{U}\widehat{\widetilde{U}}(\widehat{\widetilde{U}}\widetilde{U})^{\bigtriangleup}+
            \widetilde{U}\widehat{U}\widehat{\widetilde{U}}(\widehat{\widetilde{U}}U)^{\bigtriangleup}
            +
            \widetilde{U}\widehat{U}\widehat{\widetilde{U}}(\widehat{U}\widetilde{U})^{\bigtriangleup}+
            \widetilde{U}\widehat{U}\widehat{\widetilde{U}}(\widehat{U}U)^{\bigtriangleup}+
            \widetilde{U}\widehat{U}\widehat{\widetilde{U}}(\widehat{U}\widehat{\widetilde{U}})^{\bigtriangleup}
            +
            U(\widehat{U}\widehat{\widetilde{U}}U\widetilde{U})^{\bigtriangleup}+
            U\widehat{\widetilde{U}}(\widehat{\widetilde{U}}U\widetilde{U})^{\bigtriangleup}+
            U\widehat{\widetilde{U}}(\widehat{U}U\widetilde{U})^{\bigtriangleup}
            +
            U\widehat{\widetilde{U}}(\widehat{U}\widehat{\widetilde{U}}\widetilde{U})^{\bigtriangleup}+
            U\widehat{\widetilde{U}}(\widehat{U}\widehat{\widetilde{U}}U)^{\bigtriangleup}+
            U\widehat{U}(\widehat{\widetilde{U}}U\widetilde{U})^{\bigtriangleup}
            +
            U\widehat{U}(\widehat{U}U\widetilde{U})^{\bigtriangleup}+
            U\widehat{U}(\widehat{U}\widehat{\widetilde{U}}\widetilde{U})^{\bigtriangleup}+
            U\widehat{U}(\widehat{U}\widehat{\widetilde{U}}U)^{\bigtriangleup}
            +
            U\widehat{U}\widehat{\widetilde{U}}(U\widetilde{U})^{\bigtriangleup}+
            U\widehat{U}\widehat{\widetilde{U}}(\widehat{\widetilde{U}}\widetilde{U})^{\bigtriangleup}+
            U\widehat{U}\widehat{\widetilde{U}}(\widehat{\widetilde{U}}U)^{\bigtriangleup}
            +
            U\widehat{U}\widehat{\widetilde{U}}(\widehat{U}\widetilde{U})^{\bigtriangleup}+
            U\widehat{U}\widehat{\widetilde{U}}(\widehat{U}U)^{\bigtriangleup}+
            U\widehat{U}\widehat{\widetilde{U}}(\widehat{U}\widehat{\widetilde{U}})^{\bigtriangleup}
            +
            U\widetilde{U}(\widehat{\widetilde{U}}U\widetilde{U})^{\bigtriangleup}+
            U\widetilde{U}(\widehat{U}U\widetilde{U})^{\bigtriangleup}+
            U\widetilde{U}(\widehat{U}\widehat{\widetilde{U}}\widetilde{U})^{\bigtriangleup}
            +
            U\widetilde{U}(\widehat{U}\widehat{\widetilde{U}}U)^{\bigtriangleup}+
            U\widetilde{U}\widehat{\widetilde{U}}(U\widetilde{U})^{\bigtriangleup}+
            U\widetilde{U}\widehat{\widetilde{U}}(\widehat{\widetilde{U}}\widetilde{U})^{\bigtriangleup}
            +
            U\widetilde{U}\widehat{\widetilde{U}}(\widehat{\widetilde{U}}U)^{\bigtriangleup}+
            U\widetilde{U}\widehat{\widetilde{U}}(\widehat{U}\widetilde{U})^{\bigtriangleup}+
            U\widetilde{U}\widehat{\widetilde{U}}(\widehat{U}U)^{\bigtriangleup}
            +
            U\widetilde{U}\widehat{\widetilde{U}}(\widehat{U}\widehat{\widetilde{U}})^{\bigtriangleup}+
            U\widetilde{U}\widehat{U}(U\widetilde{U})^{\bigtriangleup}+
            U\widetilde{U}\widehat{U}(\widehat{\widetilde{U}}\widetilde{U})^{\bigtriangleup}
            +
            U\widetilde{U}\widehat{U}(\widehat{\widetilde{U}}U)^{\bigtriangleup}+
            U\widetilde{U}\widehat{U}(\widehat{U}\widetilde{U})^{\bigtriangleup}+
            U\widetilde{U}\widehat{U}(\widehat{U}U)^{\bigtriangleup}
            +
            U\widetilde{U}\widehat{U}(\widehat{U}\widehat{\widetilde{U}})^{\bigtriangleup}+
            U\widetilde{U}\widehat{U}\widehat{\widetilde{U}}\widetilde{U}^{\bigtriangleup}+
            U\widetilde{U}\widehat{U}\widehat{\widetilde{U}}U^{\bigtriangleup}\!\!
            +
            U\widetilde{U}\widehat{U}\widehat{\widetilde{U}}\widehat{\widetilde{U}}^{\bigtriangleup}+
            U\widetilde{U}\widehat{U}\widehat{\widetilde{U}}\widehat{U}^{\bigtriangleup})
            +
            \frac{2}{3}(
            (\widehat{\widetilde{U}}^{\bigtriangleup}((\widehat{U}\widehat{\widetilde{U}})^{\bigtriangleup}(U\widetilde{U})^{\bigtriangleup})^{\bigtriangleup})^{\bigtriangleup}
            +\\
            (\widehat{U}^{\bigtriangleup}((\widehat{U}\widehat{\widetilde{U}})^{\bigtriangleup}(U\widetilde{U})^{\bigtriangleup})^{\bigtriangleup})^{\bigtriangleup}
            +
            ((\widehat{U}\widehat{\widetilde{U}})^{\bigtriangleup}(\widehat{\widetilde{U}}^{\bigtriangleup}(U\widetilde{U})^{\bigtriangleup})^{\bigtriangleup})^{\bigtriangleup}
            +
            ((\widehat{U}\widehat{\widetilde{U}})^{\bigtriangleup}(\widehat{U}^{\bigtriangleup}(U\widetilde{U})^{\bigtriangleup})^{\bigtriangleup})^{\bigtriangleup}
            +\\
            ((\widehat{U}\widehat{\widetilde{U}})^{\bigtriangleup}((\widehat{U}\widehat{\widetilde{U}})^{\bigtriangleup}\widetilde{U}^{\bigtriangleup})^{\bigtriangleup})^{\bigtriangleup}+
            ((\widehat{U}\widehat{\widetilde{U}})^{\bigtriangleup}((\widehat{U}\widehat{\widetilde{U}})^{\bigtriangleup}U^{\bigtriangleup})^{\bigtriangleup})^{\bigtriangleup}
            +
            \widetilde{U}(\widehat{U}\widehat{\widetilde{U}})^{\bigtriangleup}(U\widetilde{U})^{\bigtriangleup}
            +\\
            \widetilde{U}(\widehat{\widetilde{U}}^{\bigtriangleup}(\widehat{\widetilde{U}}^{\bigtriangleup}(U\widetilde{U})^{\bigtriangleup})^{\bigtriangleup})^{\bigtriangleup}
            +
            \widetilde{U}(\widehat{\widetilde{U}}^{\bigtriangleup}(\widehat{U}^{\bigtriangleup}(U\widetilde{U})^{\bigtriangleup})^{\bigtriangleup})^{\bigtriangleup}
            +
            \widetilde{U}(\widehat{\widetilde{U}}^{\bigtriangleup}((\widehat{U}\widehat{\widetilde{U}})^{\bigtriangleup}\widetilde{U}^{\bigtriangleup})^{\bigtriangleup})^{\bigtriangleup}
            +\\
            \widetilde{U}(\widehat{\widetilde{U}}^{\bigtriangleup}((\widehat{U}\widehat{\widetilde{U}})^{\bigtriangleup}U^{\bigtriangleup})^{\bigtriangleup})^{\bigtriangleup}
            +
            \widetilde{U}(\widehat{U}^{\bigtriangleup}(\widehat{\widetilde{U}}^{\bigtriangleup}(U\widetilde{U})^{\bigtriangleup})^{\bigtriangleup})^{\bigtriangleup}
            +
            \widetilde{U}(\widehat{U}^{\bigtriangleup}(\widehat{U}^{\bigtriangleup}(U\widetilde{U})^{\bigtriangleup})^{\bigtriangleup})^{\bigtriangleup}
            +\\
            \widetilde{U}(\widehat{U}^{\bigtriangleup}((\widehat{U}\widehat{\widetilde{U}})^{\bigtriangleup}\widetilde{U}^{\bigtriangleup})^{\bigtriangleup})^{\bigtriangleup}
            +
            \widetilde{U}(\widehat{U}^{\bigtriangleup}((\widehat{U}\widehat{\widetilde{U}})^{\bigtriangleup}U^{\bigtriangleup})^{\bigtriangleup})^{\bigtriangleup}
            +
            \widetilde{U}((\widehat{U}\widehat{\widetilde{U}})^{\bigtriangleup}U\widetilde{U})^{\bigtriangleup}
            +\\
            \widetilde{U}((\widehat{U}\widehat{\widetilde{U}})^{\bigtriangleup}(\widehat{\widetilde{U}}^{\bigtriangleup}\widetilde{U}^{\bigtriangleup})^{\bigtriangleup})^{\bigtriangleup}
            +
            \widetilde{U}((\widehat{U}\widehat{\widetilde{U}})^{\bigtriangleup}(\widehat{\widetilde{U}}^{\bigtriangleup}U^{\bigtriangleup})^{\bigtriangleup})^{\bigtriangleup}
            +
            \widetilde{U}((\widehat{U}\widehat{\widetilde{U}})^{\bigtriangleup}(\widehat{U}^{\bigtriangleup}\widetilde{U}^{\bigtriangleup})^{\bigtriangleup})^{\bigtriangleup}
            +\\
            \widetilde{U}((\widehat{U}\widehat{\widetilde{U}})^{\bigtriangleup}(\widehat{U}^{\bigtriangleup}U^{\bigtriangleup})^{\bigtriangleup})^{\bigtriangleup}
            +
            \widetilde{U}((\widehat{U}\widehat{\widetilde{U}})^{\bigtriangleup}\widehat{U}\widehat{\widetilde{U}})^{\bigtriangleup}
            +
            U(\widehat{U}\widehat{\widetilde{U}})^{\bigtriangleup}(U\widetilde{U})^{\bigtriangleup}
            +
            U(\widehat{\widetilde{U}}^{\bigtriangleup}(\widehat{\widetilde{U}}^{\bigtriangleup}(U\widetilde{U})^{\bigtriangleup})^{\bigtriangleup})^{\bigtriangleup}
            +\\
            U(\widehat{\widetilde{U}}^{\bigtriangleup}(\widehat{U}^{\bigtriangleup}(U\widetilde{U})^{\bigtriangleup})^{\bigtriangleup})^{\bigtriangleup}
            +
            U(\widehat{\widetilde{U}}^{\bigtriangleup}((\widehat{U}\widehat{\widetilde{U}})^{\bigtriangleup}\widetilde{U}^{\bigtriangleup})^{\bigtriangleup})^{\bigtriangleup}
            +
            U(\widehat{\widetilde{U}}^{\bigtriangleup}((\widehat{U}\widehat{\widetilde{U}})^{\bigtriangleup}U^{\bigtriangleup})^{\bigtriangleup})^{\bigtriangleup}
            +\\
            U(\widehat{U}^{\bigtriangleup}(\widehat{\widetilde{U}}^{\bigtriangleup}(U\widetilde{U})^{\bigtriangleup})^{\bigtriangleup})^{\bigtriangleup}
            +
            U(\widehat{U}^{\bigtriangleup}(\widehat{U}^{\bigtriangleup}(U\widetilde{U})^{\bigtriangleup})^{\bigtriangleup})^{\bigtriangleup}
            +
            U(\widehat{U}^{\bigtriangleup}((\widehat{U}\widehat{\widetilde{U}})^{\bigtriangleup}\widetilde{U}^{\bigtriangleup})^{\bigtriangleup})^{\bigtriangleup}
            +\\
            U(\widehat{U}^{\bigtriangleup}((\widehat{U}\widehat{\widetilde{U}})^{\bigtriangleup}U^{\bigtriangleup})^{\bigtriangleup})^{\bigtriangleup}
            +
            U((\widehat{U}\widehat{\widetilde{U}})^{\bigtriangleup}U\widetilde{U})^{\bigtriangleup}
            +
            U((\widehat{U}\widehat{\widetilde{U}})^{\bigtriangleup}(\widehat{\widetilde{U}}^{\bigtriangleup}\widetilde{U}^{\bigtriangleup})^{\bigtriangleup})^{\bigtriangleup}
            +\\
            U((\widehat{U}\widehat{\widetilde{U}})^{\bigtriangleup}(\widehat{\widetilde{U}}^{\bigtriangleup}U^{\bigtriangleup})^{\bigtriangleup})^{\bigtriangleup}
            +
            U((\widehat{U}\widehat{\widetilde{U}})^{\bigtriangleup}(\widehat{U}^{\bigtriangleup}\widetilde{U}^{\bigtriangleup})^{\bigtriangleup})^{\bigtriangleup}
            +
            U((\widehat{U}\widehat{\widetilde{U}})^{\bigtriangleup}(\widehat{U}^{\bigtriangleup}U^{\bigtriangleup})^{\bigtriangleup})^{\bigtriangleup}
            +\\
            U((\widehat{U}\widehat{\widetilde{U}})^{\bigtriangleup}\widehat{U}\widehat{\widetilde{U}})^{\bigtriangleup}
            +
            U\widetilde{U}\widehat{\widetilde{U}}^{\bigtriangleup}(U\widetilde{U})^{\bigtriangleup}
            +
            U\widetilde{U}\widehat{U}^{\bigtriangleup}(U\widetilde{U})^{\bigtriangleup}
            +
            U\widetilde{U}(\widehat{U}\widehat{\widetilde{U}})^{\bigtriangleup}\widetilde{U}^{\bigtriangleup}
            +
            U\widetilde{U}(\widehat{U}\widehat{\widetilde{U}})^{\bigtriangleup}U^{\bigtriangleup}
            +
            U\widetilde{U}(\widehat{\widetilde{U}}^{\bigtriangleup}U\widetilde{U})^{\bigtriangleup}+
            U\widetilde{U}(\widehat{\widetilde{U}}^{\bigtriangleup}(\widehat{\widetilde{U}}^{\bigtriangleup}\widetilde{U}^{\bigtriangleup})^{\bigtriangleup})^{\bigtriangleup}
            +
            U\widetilde{U}(\widehat{\widetilde{U}}^{\bigtriangleup}(\widehat{\widetilde{U}}^{\bigtriangleup}U^{\bigtriangleup})^{\bigtriangleup})^{\bigtriangleup}
            +
            U\widetilde{U}(\widehat{\widetilde{U}}^{\bigtriangleup}(\widehat{U}^{\bigtriangleup}\widetilde{U}^{\bigtriangleup})^{\bigtriangleup})^{\bigtriangleup}
            +
            U\widetilde{U}(\widehat{\widetilde{U}}^{\bigtriangleup}(\widehat{U}^{\bigtriangleup}U^{\bigtriangleup})^{\bigtriangleup})^{\bigtriangleup}
            +
            U\widetilde{U}(\widehat{\widetilde{U}}^{\bigtriangleup}\widehat{U}\widehat{\widetilde{U}})^{\bigtriangleup}
            +
            U\widetilde{U}(\widehat{U}^{\bigtriangleup}U\widetilde{U})^{\bigtriangleup}
            +
            U\widetilde{U}(\widehat{U}^{\bigtriangleup}(\widehat{\widetilde{U}}^{\bigtriangleup}\widetilde{U}^{\bigtriangleup})^{\bigtriangleup})^{\bigtriangleup}
            +\\
            U\widetilde{U}(\widehat{U}^{\bigtriangleup}(\widehat{\widetilde{U}}^{\bigtriangleup}U^{\bigtriangleup})^{\bigtriangleup})^{\bigtriangleup}
            +
            U\widetilde{U}(\widehat{U}^{\bigtriangleup}(\widehat{U}^{\bigtriangleup}\widetilde{U}^{\bigtriangleup})^{\bigtriangleup})^{\bigtriangleup}
            +
            U\widetilde{U}(\widehat{U}^{\bigtriangleup}(\widehat{U}^{\bigtriangleup}U^{\bigtriangleup})^{\bigtriangleup})^{\bigtriangleup}+
            U\widetilde{U}(\widehat{U}^{\bigtriangleup}\widehat{U}\widehat{\widetilde{U}})^{\bigtriangleup}
            +
            U\widetilde{U}((\widehat{U}\widehat{\widetilde{U}})^{\bigtriangleup}\widetilde{U})^{\bigtriangleup}+
            U\widetilde{U}((\widehat{U}\widehat{\widetilde{U}})^{\bigtriangleup}U)^{\bigtriangleup}
            +
            U\widetilde{U}((\widehat{U}\widehat{\widetilde{U}})^{\bigtriangleup}\widehat{\widetilde{U}})^{\bigtriangleup}+
            U\widetilde{U}((\widehat{U}\widehat{\widetilde{U}})^{\bigtriangleup}\widehat{U})^{\bigtriangleup}),
            $
            \\\\
        $
            C_{(6)}=-(\frac{1}{3}(\widehat{U}\widehat{\widetilde{U}}(\widehat{U}\widehat{\widetilde{U}}U\widetilde{U})^{\bigtriangleup}+
            \widetilde{U}\widehat{\widetilde{U}}(\widehat{U}\widehat{\widetilde{U}}U\widetilde{U})^{\bigtriangleup}+
            \widetilde{U}\widehat{U}(\widehat{U}\widehat{\widetilde{U}}U\widetilde{U})^{\bigtriangleup}
            +
            \widetilde{U}\widehat{U}\widehat{\widetilde{U}}(\widehat{\widetilde{U}}U\widetilde{U})^{\bigtriangleup}+\\
            \widetilde{U}\widehat{U}\widehat{\widetilde{U}}(\widehat{U}U\widetilde{U})^{\bigtriangleup}+
            \widetilde{U}\widehat{U}\widehat{\widetilde{U}}(\widehat{U}\widehat{\widetilde{U}}\widetilde{U})^{\bigtriangleup}
            +
            \widetilde{U}\widehat{U}\widehat{\widetilde{U}}(\widehat{U}\widehat{\widetilde{U}}U)^{\bigtriangleup}+
            U\widehat{\widetilde{U}}(\widehat{U}\widehat{\widetilde{U}}U\widetilde{U})^{\bigtriangleup}+
            U\widehat{U}(\widehat{U}\widehat{\widetilde{U}}U\widetilde{U})^{\bigtriangleup}
            +
            U\widehat{U}\widehat{\widetilde{U}}(\widehat{\widetilde{U}}U\widetilde{U})^{\bigtriangleup}+
            U\widehat{U}\widehat{\widetilde{U}}(\widehat{U}U\widetilde{U})^{\bigtriangleup}+
            U\widehat{U}\widehat{\widetilde{U}}(\widehat{U}\widehat{\widetilde{U}}\widetilde{U})^{\bigtriangleup}
            +
            U\widehat{U}\widehat{\widetilde{U}}(\widehat{U}\widehat{\widetilde{U}}U)^{\bigtriangleup}+
            U\widetilde{U}(\widehat{U}\widehat{\widetilde{U}}U\widetilde{U})^{\bigtriangleup}+
            U\widetilde{U}\widehat{\widetilde{U}}(\widehat{\widetilde{U}}U\widetilde{U})^{\bigtriangleup}
            +
            U\widetilde{U}\widehat{\widetilde{U}}(\widehat{U}U\widetilde{U})^{\bigtriangleup}+
            U\widetilde{U}\widehat{\widetilde{U}}(\widehat{U}\widehat{\widetilde{U}}\widetilde{U})^{\bigtriangleup}+
            U\widetilde{U}\widehat{\widetilde{U}}(\widehat{U}\widehat{\widetilde{U}}U)^{\bigtriangleup}
            +
            U\widetilde{U}\widehat{U}(\widehat{\widetilde{U}}U\widetilde{U})^{\bigtriangleup}+
            U\widetilde{U}\widehat{U}(\widehat{U}U\widetilde{U})^{\bigtriangleup}+
            U\widetilde{U}\widehat{U}(\widehat{U}\widehat{\widetilde{U}}\widetilde{U})^{\bigtriangleup}
            +
            U\widetilde{U}\widehat{U}(\widehat{U}\widehat{\widetilde{U}}U)^{\bigtriangleup}+
            U\widetilde{U}\widehat{U}\widehat{\widetilde{U}}(U\widetilde{U})^{\bigtriangleup}+
            U\widetilde{U}\widehat{U}\widehat{\widetilde{U}}(\widehat{\widetilde{U}}\widetilde{U})^{\bigtriangleup}
            +\\
            U\widetilde{U}\widehat{U}\widehat{\widetilde{U}}(\widehat{\widetilde{U}}U)^{\bigtriangleup}+
            U\widetilde{U}\widehat{U}\widehat{\widetilde{U}}(\widehat{U}\widetilde{U})^{\bigtriangleup}+
            U\widetilde{U}\widehat{U}\widehat{\widetilde{U}}(\widehat{U}U)^{\bigtriangleup}
            +
            U\widetilde{U}\widehat{U}\widehat{\widetilde{U}}(\widehat{U}\widehat{\widetilde{U}})^{\bigtriangleup})+\\
            \frac{2}{3}(
            ((\widehat{U}\widehat{\widetilde{U}})^{\bigtriangleup}((\widehat{U}\widehat{\widetilde{U}})^{\bigtriangleup}(U\widetilde{U})^{\bigtriangleup})^{\bigtriangleup})^{\bigtriangleup}\!+
            \widetilde{U}(\widehat{\widetilde{U}}^{\bigtriangleup}\!\!\!((\widehat{U}\widehat{\widetilde{U}})^{\bigtriangleup}(U\widetilde{U})^{\bigtriangleup})^{\bigtriangleup})^{\bigtriangleup}
            \!+
            \widetilde{U}(\widehat{U}^{\bigtriangleup}((\widehat{U}\widehat{\widetilde{U}})^{\bigtriangleup}(U\widetilde{U})^{\bigtriangleup})^{\bigtriangleup})^{\bigtriangleup}+
            \widetilde{U}((\widehat{U}\widehat{\widetilde{U}})^{\bigtriangleup}(\widehat{\widetilde{U}}^{\bigtriangleup}(U\widetilde{U})^{\bigtriangleup})^{\bigtriangleup})^{\bigtriangleup}
            +
            \widetilde{U}((\widehat{U}\widehat{\widetilde{U}})^{\bigtriangleup}(\widehat{U}^{\bigtriangleup}(U\widetilde{U})^{\bigtriangleup})^{\bigtriangleup})^{\bigtriangleup}+
            \widetilde{U}((\widehat{U}\widehat{\widetilde{U}})^{\bigtriangleup}((\widehat{U}\widehat{\widetilde{U}})^{\bigtriangleup}\widetilde{U}^{\bigtriangleup})^{\bigtriangleup})^{\bigtriangleup}
            +
            \widetilde{U}((\widehat{U}\widehat{\widetilde{U}})^{\bigtriangleup}((\widehat{U}\widehat{\widetilde{U}})^{\bigtriangleup}U^{\bigtriangleup})^{\bigtriangleup})^{\bigtriangleup}+
            U(\widehat{\widetilde{U}}^{\bigtriangleup}((\widehat{U}\widehat{\widetilde{U}})^{\bigtriangleup}(U\widetilde{U})^{\bigtriangleup})^{\bigtriangleup})^{\bigtriangleup}
            +
            U(\widehat{U}^{\bigtriangleup}((\widehat{U}\widehat{\widetilde{U}})^{\bigtriangleup}(U\widetilde{U})^{\bigtriangleup})^{\bigtriangleup})^{\bigtriangleup}+
            U((\widehat{U}\widehat{\widetilde{U}})^{\bigtriangleup}(\widehat{\widetilde{U}}^{\bigtriangleup}(U\widetilde{U})^{\bigtriangleup})^{\bigtriangleup})^{\bigtriangleup}
            +
            U((\widehat{U}\widehat{\widetilde{U}})^{\bigtriangleup}(\widehat{U}^{\bigtriangleup}(U\widetilde{U})^{\bigtriangleup})^{\bigtriangleup})^{\bigtriangleup}+
            U((\widehat{U}\widehat{\widetilde{U}})^{\bigtriangleup}((\widehat{U}\widehat{\widetilde{U}})^{\bigtriangleup}\widetilde{U}^{\bigtriangleup})^{\bigtriangleup})^{\bigtriangleup}
            +
            U((\widehat{U}\widehat{\widetilde{U}})^{\bigtriangleup}((\widehat{U}\widehat{\widetilde{U}})^{\bigtriangleup}U^{\bigtriangleup})^{\bigtriangleup})^{\bigtriangleup}+
            U\widetilde{U}(\widehat{U}\widehat{\widetilde{U}})^{\bigtriangleup}(U\widetilde{U})^{\bigtriangleup}
            +
            U\widetilde{U}(\widehat{\widetilde{U}}^{\bigtriangleup}(\widehat{\widetilde{U}}^{\bigtriangleup}(U\widetilde{U})^{\bigtriangleup})^{\bigtriangleup})^{\bigtriangleup}+\\
            U\widetilde{U}(\widehat{\widetilde{U}}^{\bigtriangleup}(\widehat{U}^{\bigtriangleup}(U\widetilde{U})^{\bigtriangleup})^{\bigtriangleup})^{\bigtriangleup}
            +
            U\widetilde{U}(\widehat{\widetilde{U}}^{\bigtriangleup}((\widehat{U}\widehat{\widetilde{U}})^{\bigtriangleup}\widetilde{U}^{\bigtriangleup})^{\bigtriangleup})^{\bigtriangleup}+
            U\widetilde{U}(\widehat{\widetilde{U}}^{\bigtriangleup}((\widehat{U}\widehat{\widetilde{U}})^{\bigtriangleup}U^{\bigtriangleup})^{\bigtriangleup})^{\bigtriangleup}
            +\\
            U\widetilde{U}(\widehat{U}^{\bigtriangleup}(\widehat{\widetilde{U}}^{\bigtriangleup}(U\widetilde{U})^{\bigtriangleup})^{\bigtriangleup})^{\bigtriangleup}+
            U\widetilde{U}(\widehat{U}^{\bigtriangleup}(\widehat{U}^{\bigtriangleup}(U\widetilde{U})^{\bigtriangleup})^{\bigtriangleup})^{\bigtriangleup}
            +
            U\widetilde{U}(\widehat{U}^{\bigtriangleup}((\widehat{U}\widehat{\widetilde{U}})^{\bigtriangleup}\widetilde{U}^{\bigtriangleup})^{\bigtriangleup})^{\bigtriangleup}+\\
            U\widetilde{U}(\widehat{U}^{\bigtriangleup}((\widehat{U}\widehat{\widetilde{U}})^{\bigtriangleup}U^{\bigtriangleup})^{\bigtriangleup})^{\bigtriangleup}
            +
            U\widetilde{U}((\widehat{U}\widehat{\widetilde{U}})^{\bigtriangleup}U\widetilde{U})^{\bigtriangleup}+
            U\widetilde{U}((\widehat{U}\widehat{\widetilde{U}})^{\bigtriangleup}(\widehat{\widetilde{U}}^{\bigtriangleup}\widetilde{U}^{\bigtriangleup})^{\bigtriangleup})^{\bigtriangleup}
            +\\
            U\widetilde{U}((\widehat{U}\widehat{\widetilde{U}})^{\bigtriangleup}(\widehat{\widetilde{U}}^{\bigtriangleup}U^{\bigtriangleup})^{\bigtriangleup})^{\bigtriangleup}+
            U\widetilde{U}((\widehat{U}\widehat{\widetilde{U}})^{\bigtriangleup}(\widehat{U}^{\bigtriangleup}\widetilde{U}^{\bigtriangleup})^{\bigtriangleup})^{\bigtriangleup}
            +
            U\widetilde{U}((\widehat{U}\widehat{\widetilde{U}})^{\bigtriangleup}(\widehat{U}^{\bigtriangleup}U^{\bigtriangleup})^{\bigtriangleup})^{\bigtriangleup}+\\
            U\widetilde{U}((\widehat{U}\widehat{\widetilde{U}})^{\bigtriangleup}\widehat{U}\widehat{\widetilde{U}})^{\bigtriangleup}))
       $
    }

\begin{thebibliography}{99}

\bibitem{article-short_version-cgi2021} Abdulkhaev, K., Shirokov, D.: On Explicit Formulas for Characteristic Polynomial Coefficients in Geometric Algebras. In: Magnenat-Thalmann N. et al. (eds) Advances in Computer Graphics. CGI 2021. Lecture Notes in Computer Science, vol 13002. Springer, Cham. (2021). https://doi.org/10.1007/978-3-030-89029-2\_50

\bibitem{AcusPack} Acus, A., Dargys, A.: Geometric Algebra Mathematica package (2017). https://github.com/ArturasAcus/GeometricAlgebra

\bibitem{Acus} Acus, A., Dargys, A.: The Inverse of a Multivector: Beyond the Threshold $p+q=5$. Adv. Appl. Clifford Algebras \textbf{28}, 65 (2018)

\bibitem{Bayro} Bayro-Corrochano, E.: Geometric Algebra Applications, vol. I. Springer, Berlin (2019)

\bibitem{Breuils} Breuils, S., Tachibana, K. \& Hitzer, E. New Applications of Clifford’s Geometric Algebra. Adv. Appl. Clifford Algebras \textbf{32}, 17 (2022). https://doi.org/10.1007/s00006-021-01196-7

\bibitem{eigenface3} Cendrillon, R., Lovell, B.: Real-time face recognition using eigenfaces. In: Visual Communications and Image Processing. pp. 269--276 (2000). https://doi.org/10.1117/12.386642

\bibitem{Dirac} Dirac, P.: Wave Equations in Conformal Space. Annals of Mathematics, Second Series, \textbf{37}(2), 429--442 (1936). https://doi.org/10.2307/1968455

\bibitem{Doran} Doran, C., Lasenby, A.: Geometric Algebra for Physicists. Cambridge University Press, Cambridge (2003)

\bibitem{Dorst2} Dorst, L.: 3d oriented projective geometry through versors of $R^{3,3}$. Adv. Appl. Clifford Algebras \textbf{26}(4), 1137–-1172 (2016)

\bibitem{Dorst} Dorst, L., Fontijne, D., Mann, D.: Geometric Algebra for Computer Science, The Morgan Kaufmann Series in Computer Graphics, San Francisco (2007)


\bibitem{Python} Hadfield, H., Wieser, E., Arsenovic, A., Kern, R., and {The Pygae Team}: pygae/clifford: v1.3.1 (2020). https://github.com/pygae/clifford, https://doi.org/10.5281/zenodo.1453978

\bibitem{Helm} Helmstetter, J.: Characteristic polynomials in Clifford algebras and in more general algebras. Adv. Appl. Clifford Algebras \textbf{29}, 30 (2019)

\bibitem{Hestenes} Hestenes, D.: Space-Time Algebra. Gordon and Breach, New York (1966)

\bibitem{Hildenbrand} Hildenbrand, D.: Foundations of Geometric Algebra Computing. Springer, Berlin (2013)

\bibitem{Hitzer} Hitzer, E., Sangwine, S.: Multivector and multivector matrix inverses in real Clifford algebras. Applied Mathematics and Computation \textbf{311}, pp. 375--389 (2017)

\bibitem{Klawitter}  Klawitter, D.: A Clifford algebraic approach to line geometry. Adv. Appl. Clifford Algebra \textbf{24}, 713–-736 (2014)

\bibitem{Lasenby} Lasenby, A., Lasenby, J.: Applications of geometric algebra in physics and links with engineering. In: Geometric Algebra with Applications in Science and Engineering. Birkhauser, Boston (2001)

\bibitem{Li} Li, H.: Invariant Algebras and Geometric Reasoning, World Scientific (2008)

\bibitem{Lounesto} Lounesto, P.: Clifford Algebras and Spinors. Cambridge Univ. Press, Cambridge (1997)

\bibitem{HitzerMatlab} Sangwine, S., Hitzer, E.: Clifford multivector toolbox (for MATLAB), 2015--2016, (Available at: http://clifford-multivector-toolbox.sourceforge.net/)

\bibitem{Shirokov_2009} Shirokov, D.: A Classification of Lie Algebras of Pseudo-Unitary Groups in the Techniques of Clifford Algebras. Advances in Applied Clifford Algebras \textbf{20}(2), 411--425 (2010). https://doi.org/10.1007/s00006-009-0177-0

\bibitem{CGI20_extend} Shirokov, D.: Basis-free solution to Sylvester equation in Clifford algebra of arbitrary dimension. Adv. Appl. Clifford Algebras 31, 70 (2021). https://doi.org/10.1007/s00006-021-01173-0

\bibitem{quat2}	Shirokov, D.: Development of the method of quaternion typification of Clifford algebra elements. Advances in Applied Clifford Algebras \textbf{22}(2), 483--497 (2012). https://doi.org/10.1007/s00006-011-0304-6

\bibitem{CGI20} Shirokov, D.: On Basis-Free Solution to Sylvester Equation in Geometric Algebra. In: Magnenat-Thalmann N. et al. (eds) Advances in Computer Graphics. CGI 2020. Lecture Notes in Computer Science, vol 12221, pp. 541--548. Springer, Cham. (2020). https://doi.org/10.1007/978-3-030-61864-3\_46

\bibitem{Shirokov} Shirokov, D.: On computing the determinant, other characteristic polynomial coefficients, and inverse in Clifford algebras of arbitrary dimension. Computational and Applied Mathematics \textbf{40}, 173, 29 pp. (2021). https://doi.org/10.1007/s40314-021-01536-0

\bibitem{quat} Shirokov, D.: Quaternion typification of Clifford algebra elements. Advances in Applied Clifford Algebras \textbf{22}(1), 243--256 (2012). https://doi.org/10.1007/s00006-011-0288-2

\bibitem{Python2} The Pygae Team: pygae/galgebra: v0.5.0 (2020). https://github.com/pygae/galgebra

\bibitem{eigenface1} Turk, M., Pentland, A.: Face recognition using eigenfaces. In: Proc. IEEE Conference on Computer Vision and Pattern Recognition, pp. 586--591 (1991)

\bibitem{eigenface2} Turk, M., Pentland, A.: Eigenfaces for recognition. Journal of Cognitive Neuroscience \textbf{3}(1), 71--86 (1991). https://doi.org/10.1162/jocn.1991.3.1.71
\end{thebibliography}
\end{document}